\documentclass[11pt]{article}

\usepackage{amsmath,amssymb,bm}
\usepackage{cite,subfigure,graphicx}
\usepackage{color}
\usepackage{fullpage}

\usepackage{datetime}
\settimeformat{ampmtime}
\usdate
\usepackage{scrtime} 

\newcommand{\R}{\mathbb{R}}
\newcommand{\C}{\mathbb{C}}

\renewcommand{\Re}[1]{\operatorname{Re}\left(#1\right)}
\renewcommand{\Im}[1]{\operatorname{Im}\left(#1\right)}

\renewcommand{\P}[1]{\operatorname{P}\left(#1\right)}
\newcommand{\E}{\operatorname{E}}

\newcommand{\e}{e}

\newcommand{\vct}[1]{\boldsymbol{#1}}
\newcommand{\mtx}[1]{\boldsymbol{#1}}
\newcommand{\<}{\langle}
\renewcommand{\>}{\rangle}

\newcommand{\T}{\mathrm{T}}

\newcommand{\trace}{\operatorname{trace}}

\newcommand{\calE}{\mathcal{E}}
\newcommand{\calH}{\mathcal{H}}
\newcommand{\calP}{\mathcal{P}}
\newcommand{\calQ}{\mathcal{Q}}
\newcommand{\calS}{\mathcal{S}}
\newcommand{\calT}{\mathcal{T}}

\newcommand{\vg}{\vct{g}}
\newcommand{\vh}{\vct{h}}

\newcommand{\vx}{\vct{x}}
\newcommand{\vy}{\vct{y}}

\newcommand{\vzero}{\vct{0}}
\newcommand{\vphi}{\vct{\phi}}
\newcommand{\vvarphi}{\vct{\varphi}}

\newcommand{\mA}{\mtx{A}}
\newcommand{\mD}{\mtx{D}}
\newcommand{\mG}{\mtx{G}}
\newcommand{\mH}{\mtx{H}}

\newcommand{\mP}{\mtx{P}}

\newcommand{\mS}{\mtx{S}}
\newcommand{\mT}{\mtx{T}}
\newcommand{\mU}{\mtx{U}}
\newcommand{\mV}{\mtx{V}}

\newcommand{\mX}{\mtx{X}}

\newcommand{\mZ}{\mtx{Z}}

\newcommand{\mPhi}{\mtx{\Phi}}

\newcommand{\mId}{{\bf I}}
\newcommand{\mzero}{\mtx{0}}

\newcommand{\tha}{{\theta_1}}
\newcommand{\thb}{{\theta_2}}
\newcommand{\Pth}{\mP_\theta}
\newcommand{\Ptha}{\mP_{\theta_1}}
\newcommand{\Pthb}{\mP_{\theta_2}}

\newcommand{\Hth}{\mH_\theta}
\newcommand{\Vth}{\mV_\theta}

\newtheorem{lemma}{Lemma}
\newtheorem{theorem}{Theorem}
\newtheorem{definition}{Definition}
\newtheorem{proposition}{Proposition}
\newenvironment{proof}{\noindent {\bf Proof} }{\endprf\par}
\def \endprf{\hfill {\vrule height6pt width6pt depth0pt}\medskip}

\title{Compressed Subspace Matching on the Continuum}
\author{William Mantzel and Justin Romberg\thanks{W.M.\ is at Qualcomm in San Diego, CA; J.R. is in School of Electrical and Computer Engineering at Georgia Tech in Atlanta, GA.  Email: wmantzel@gmail.com, jrom@ece.gatech.edu.  This work was supported by ONR grant N00014-11-1-0459 and a grant from the Packard Foundation. 
}}

\begin{document}

\maketitle

\begin{abstract}
	
	We consider the general problem of matching a subspace to a signal in $\R^N$ that has been observed indirectly (compressed) through a random projection.  We are interested in the case where the collection of $K$-dimensional subspaces is continuously parameterized, i.e.\ naturally indexed by an interval from the real line, or more generally a region of $\R^D$.  Our main results show that if the dimension of the random projection is on the order of $K$ times a geometrical constant that describes the complexity of the collection, then the match obtained from the compressed observation is nearly as good as one obtained from a full observation of the signal.  We give multiple concrete examples of collections of subspaces for which this geometrical constant can be estimated, and discuss the relevance of the results to the general problems of template matching and source localization.
	
\end{abstract}

\section{Introduction}


We consider the general problem of finding the subspace that best approximates a fixed signal $\vh_0\in\R^N$ from a parameterized collection of $K$-dimensional subspaces  $\{\calS_\theta : \theta\in\Theta\}$, where $\theta$ is a parameter vector chosen from some compact parameter set $\Theta \subset \R^D$.  Given this collection, 
the subspace best matched to $\vh_0$ is the solution to 
\begin{equation}
	\label{eq:sm}
	\bar\theta = \arg\min_{\theta\in\Theta}\min_{\vh\in\calS_\theta}\|\vh_0-\vh\|^2_2 = 
	\arg\min_{\theta\in\Theta} \|\vh_0-\Pth\vh_0\|_2^2.
\end{equation}
The operator $\Pth$ above is the orthogonal projection onto subspace $\calS_\theta$.  In words, program \eqref{eq:sm} returns the (index of the) subspace in the collection which contains the closest point to $\vh_0$ in terms of the standard Euclidean distance.

In this paper, we explore how effectively this matching can be done from a set of indirect linear observations of $\vh_0$.  In particular, we will quantify how effectively this problem can be solved when our observations $\vy=\mPhi\vh_0$ are {\em compressed samples}, meaning $\mPhi$ is a $M\times N$ {\em underdetermined} matrix whose rows are diverse.  For simplicity, we will consider the case where the entries of $\mPhi$ are independent and Gaussian, but we expect that the majority of our results could be extended to other types of measurement scenarios.

Given measurements $\vy=\mPhi\vh_0$, we match a subspace using a variation of \eqref{eq:sm}:
\begin{equation}
	\label{eq:csm}
	\hat\theta = \arg\min_{\theta\in\Theta}\min_{\vh\in\calS_\theta} \|\vy-\mPhi\vh\|^2_2 =
	\arg\min_{\theta\in\Theta} \|\vy - \tilde\Pth\vy\|_2^2,
\end{equation}
where $\tilde\Pth$ is the orthogonal projection onto the range of $\mPhi\Pth$, the $K$ dimensional subspace of $\R^M$ consisting of measurements induced by signals in $\calS_\theta$.
We would like the compressed estimate produced by \eqref{eq:csm} to be close to the standard estimate in \eqref{eq:sm}.  We will judge the difference between these two estimates based on how well the subspaces they return can approximate the original signal $\vh_0$.  Our main result bounds the difference in the relative approximation errors $\hat{E}^2 - \bar{E}^2$, where
\begin{align*}
	\bar{E}^2 &= \frac{\|\vh_0- \mP_{\bar\theta}\vh_0\|_2^2}{\|\vh_0\|^2_2} 
	= 1 - \frac{\|\mP_{\bar\theta}\vh_0\|_2^2}{\|\vh_0\|^2_2}, \quad\text{and}\\
	\hat{E}^2 &= \frac{\|\vh_0- \mP_{\hat\theta}\vh_0\|_2^2}{\|\vh_0\|^2_2} 
	= 1 - \frac{\|\mP_{\hat\theta}\vh_0\|_2^2}{\|\vh_0\|^2_2}.
\end{align*}
Since all of the terms above scale with the size of $\vh_0$, we will assume (without loss of generality) from this point forward that $\|\vh_0\|_2=1$, and derive a bound for the gap
\begin{equation}
	\label{eq:Ediff}
	\hat{E}^2 - \bar{E}^2 = \|\mP_{\bar\theta}\vh_0\|_2^2 - \|\mP_{\hat\theta}\vh_0\|_2^2.
\end{equation}
Notice that the difference above must be positive, as $\bar\theta$ as given by \eqref{eq:sm} is the index for the optimal subspace.  Note also that we are making no assumptions about whether or not $\vh_0$ is in or even close to one of the $\calS_\theta$.
Bounding \eqref{eq:Ediff} also does not give us any immediate guarantees on how close the indices $\bar\theta$ and $\tilde\theta$ are to one another; we are only judging the effectiveness of the compressed program by its ability to produce a subspace which can estimate $\vh_0$ almost as well as the optimal one.  
Given some assumptions about the collection of subspaces (say, that $\|\Ptha-\Pthb\|$ is small only when $\|\theta_1-\theta_2\|_2$ is small), it may be possible to infer a bound on the size of $\hat\theta-\tilde\theta$ from a bound on \eqref{eq:Ediff}, but we do not explore this here.

Our main results, stated as Theorems~\ref{th:supW}--\ref{th:suphperp} in Section~\ref{sec:main-results}, provide a bound on this approximation gap whose most significant term has the form
\begin{equation}
	\label{eq:Ebound1}
	\hat{E}^2 - \bar{E}^2 ~\lesssim~ \sqrt{\frac{K(\Delta+\log K)}{M}},
\end{equation}
where $\Delta$ is a quantity that captures the geometrical complexity of the collection of subspaces $\{\calS_\theta\}$.  As described in detail in Section~\ref{sec:main-results}, it is related to the covering number of $\{\calS_\theta\}$ under the standard metric for subspaces, the operator norm of the difference of projectors $d(\theta_1,\theta_2) = \|\Ptha-\Pthb\|$.  In all of our motivating applications, we will have $\Delta = \text{small constant} + \log(K)$.  This means that we can control the approximation gap in \eqref{eq:Ediff} by making $M$ slightly larger than $K$ (and potentially significantly smaller than $N$).  The result \eqref{eq:Ebound1} holds ``with high probability'' for an arbitrary (but fixed) $\vh_0$, where the randomness comes from the matrix $\mPhi$.

As this requires solving a least-squares problem of dimension $K$, it will only be practical to search over all of $\Theta$ when $D$ (the dimension of the $\theta\in\Theta$) is small ($1$ or $2$).  This is indeed the case in the applications we discuss in the next section.    

\subsection{Motivating applications}
\label{sec:motivating-applications}


\subsubsection{Pulse detection from compressed samples}

Perhaps the most fundamental example of a subspace matching problem is estimating the time shift (or ``time of arrival'') of a known pulse signal.  In this case, the family $\calS_\theta$ consists of one dimensional subspaces corresponding to scalings of a shifted template function $f_0(t)$; $\calS_\theta = \{a\, f_0(t-\theta) ~ : ~ a \in \R\}$.  Treatment of this problem with observations made through a random matrix $\mPhi$ was first given in \cite{davenport07sm,davenport10si}, where it was called the {\em smashed filter}.  An analysis of this of problem with a $\mPhi$ corresponding to measurements made in the frequency domain can be found in \cite{eftekhari13ma}.

Figure~\ref{fig:cGabor} illustrates a variation on this problem for which we give a thorough analysis in Section~\ref{sec:gabor}.  Here we are trying to find the shift and modulation frequency of a pulse of fixed shape that most closely matches the observed signal.  The family of subspaces is indexed by the multi-parameter $\theta=(\tau,\omega)$, and consists of finely sampled\footnote{As we will discuss further in Section~\ref{sec:geometric-regularity}, the number of samples $N$ plays no role in our analysis.  We will simply assume that is is large enough so that the relationships between the $g(t;\cdot)$ are preserved.} Gabor functions at different shifts, frequencies, and phases.  We define
\begin{equation}
	\label{eq:gaborct}
	g(t;\tau,\omega,a,b) = e^{-(t-\tau)^2/\sigma^2}\left(a\cos(\omega t) + b\sin(\omega t)\right)
\end{equation}
and $g[n;\tau,\omega,a,b] = g(t_1+n\,d;\tau,\omega,a,b)$ for some sample spacing $d$ and $n=1,\ldots,N$; the vector $g[n;\cdot]$ contains equally spaced samples of $g(t;\cdot)$ on the interval $[t_1,t_2]$ with $t_2=t_1+Nd$.  We can now write  
\begin{equation}
	\label{eq:gaborSth}
	\calS_\theta = \{g[n;\tau,\omega,a,b]~a,b\in\R,~(\tau,\omega)=\theta\in\Theta\}.
\end{equation}
Each subspace has a dimension of $K=2$, and $\Theta$ is some subset of the time-frequency plane; here we will fix the specific region as $\Theta = [-0.25,0.25]\times[2\pi\cdot 50,2\pi\cdot 250]$.  The example in Figure~\ref{fig:cGabor} uses a $\vh_0$ consisting of (again finely spaced) samples of the signal $h(t) = (1+\cos(\pi t/\beta_0))\cdot\cos(2\pi\cdot5 t/\beta_0 + \pi/3)\cdot 1_{|t|\leq\beta_0}$ where $\beta_0 = 5/128$.  Notice that $\vh_0$ is not actually included in any of the subspaces in $\calS_\theta$; but the closest match does indeed comes at the expected values of $\hat\tau=0$ and $\hat\omega=2\pi\cdot 128$.

We see from the figure that gap between the uncompressed energy surface $\|\vh_0-\Pth\vh_0\|_2^2$ and the compressed energy surface $\|\vy-\tilde\Pth\vy\|_2^2$ gets tighter as $M$ increases.  This is mathematically characterized by our main theorems along with the computations in Section~\ref{sec:gabor}.

\begin{figure}[ht]
	\centering
	\begin{tabular}{cc}
		\includegraphics[width=2in]{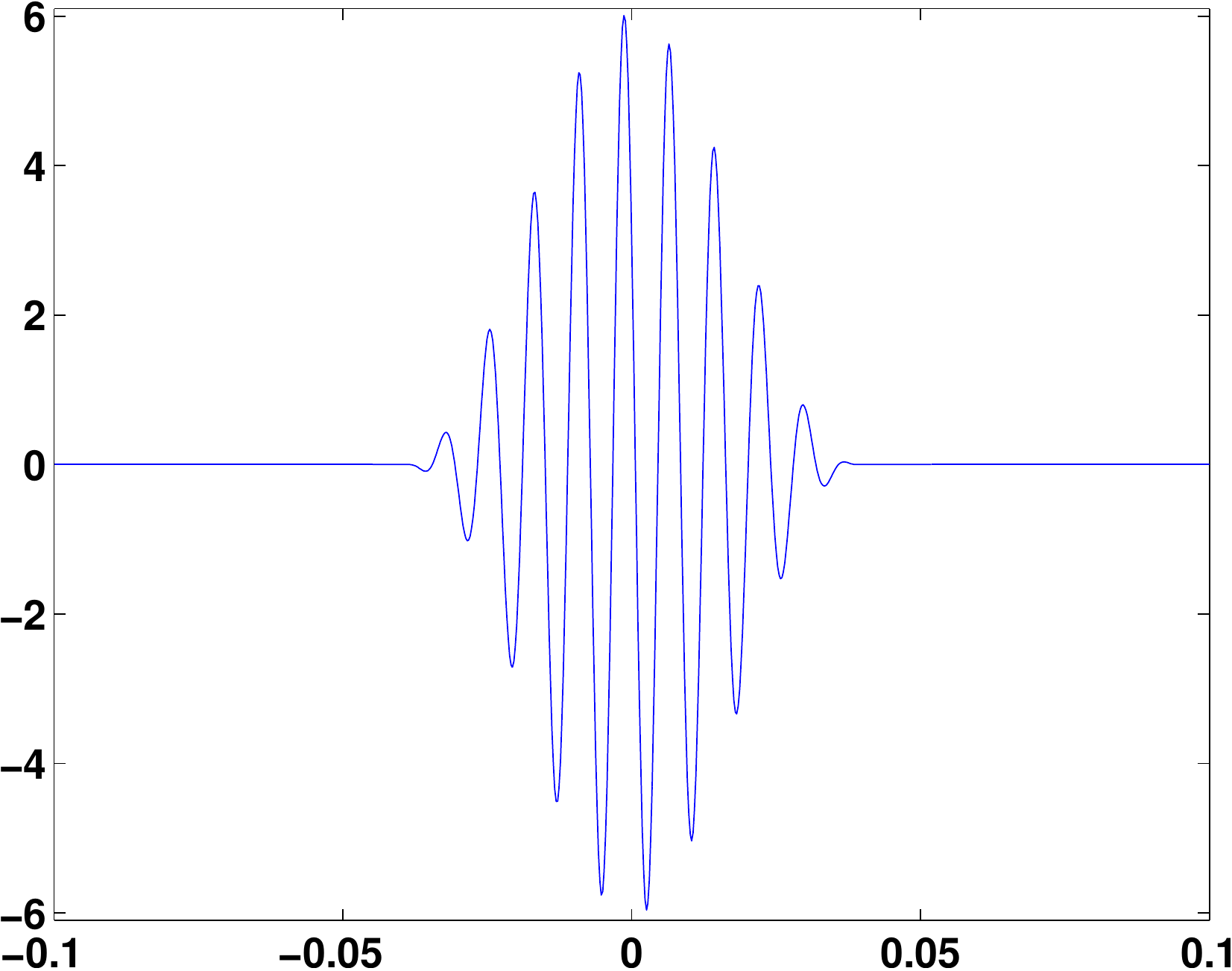} &
		\hspace{.1in} 
		\includegraphics[width=2in]{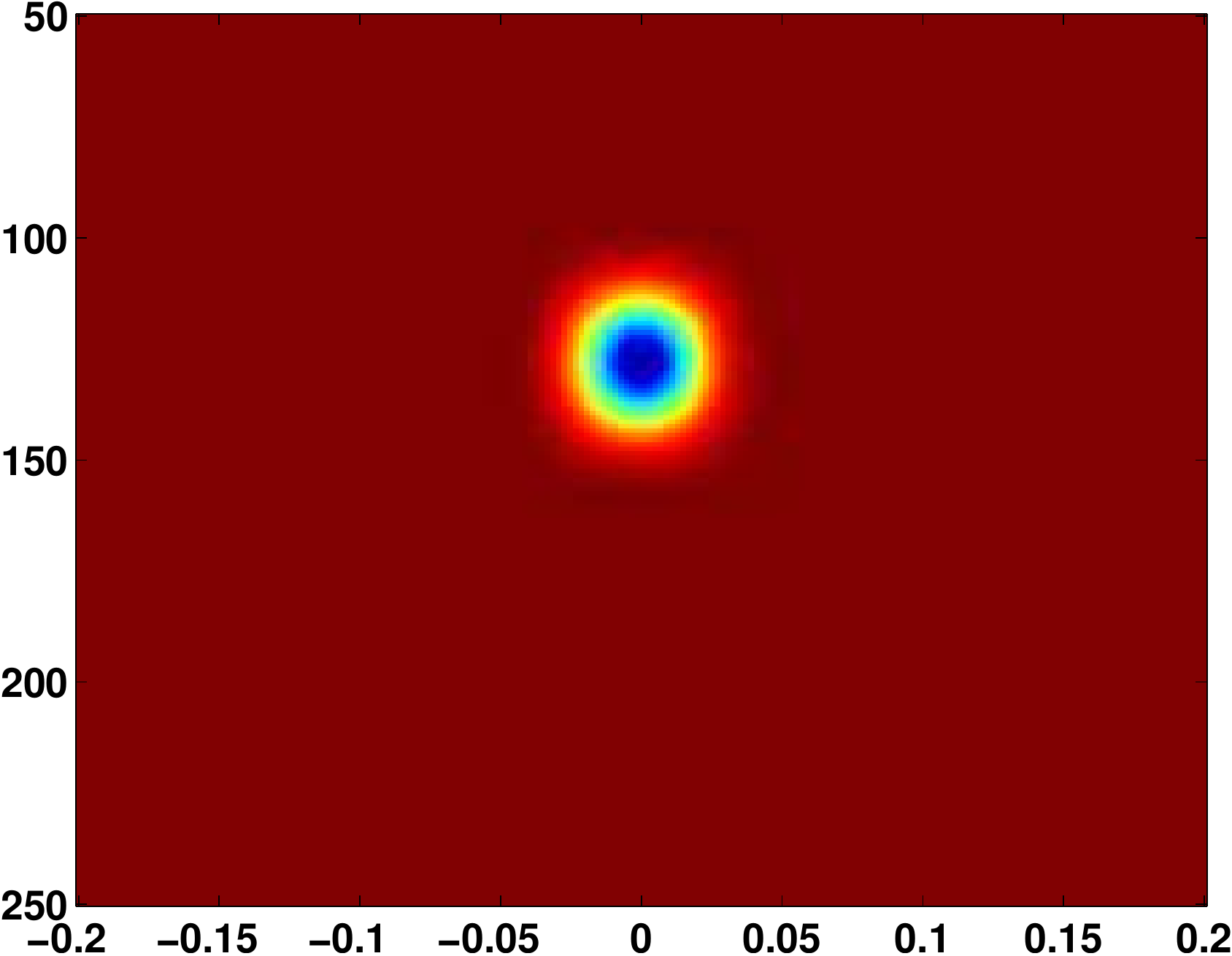} \\
		(a) & (b)
	\end{tabular}
	
	\vspace{.25in}
	
	\begin{tabular}{ccc} 
		\includegraphics[width=2in]{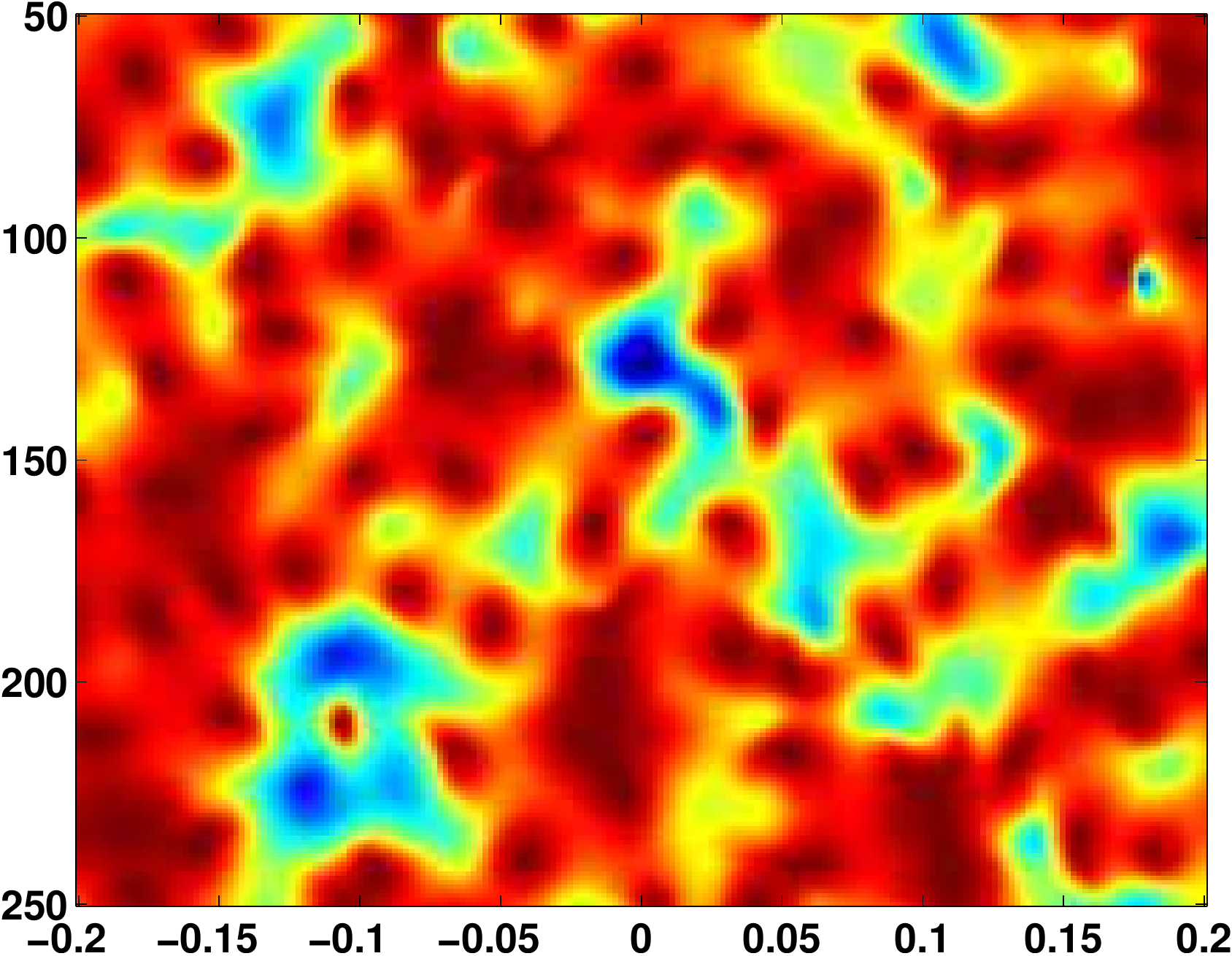} &
		\hspace{.1in} 
		\includegraphics[width=2in]{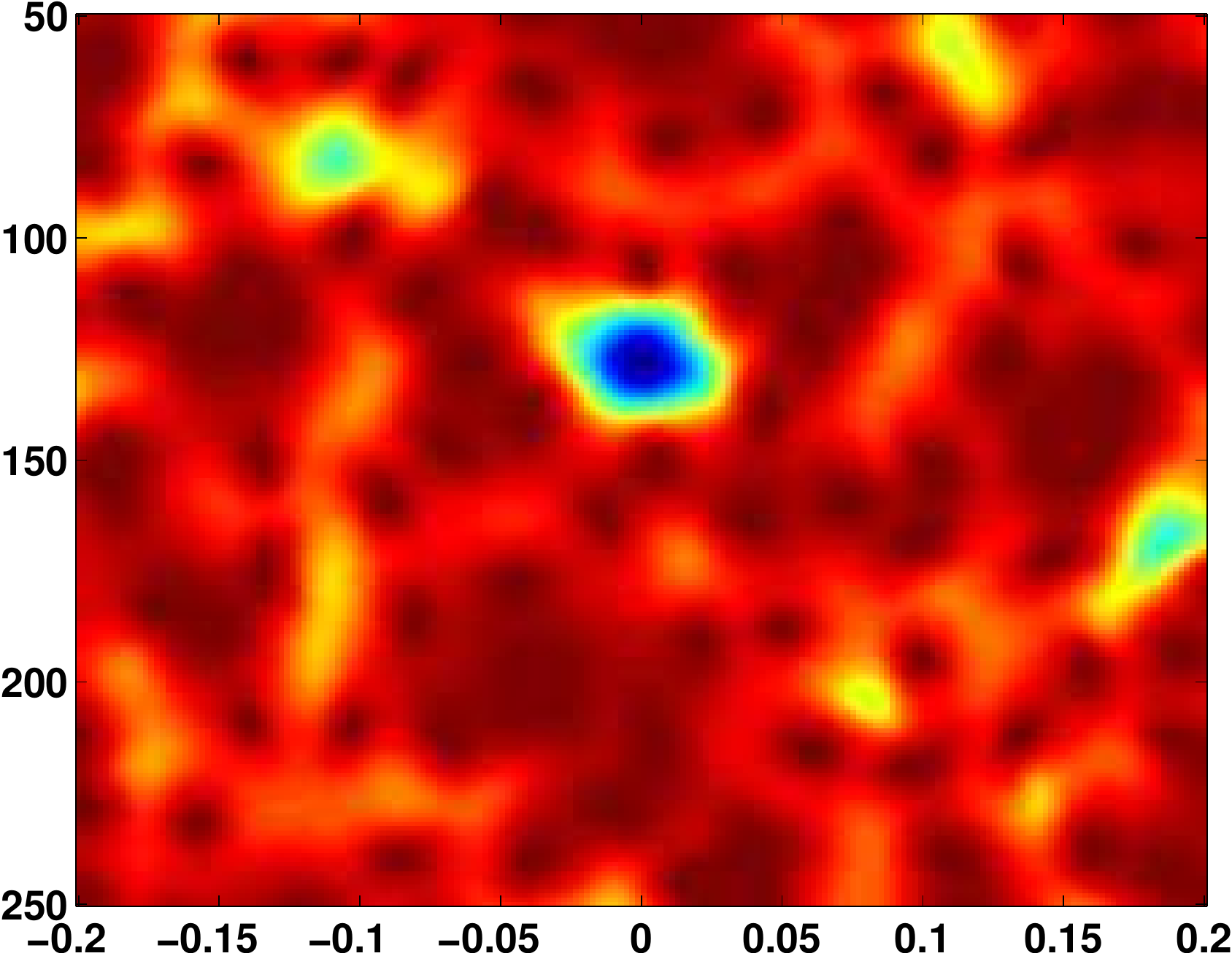} &
		\hspace{.1in} 
		\includegraphics[width=2in]{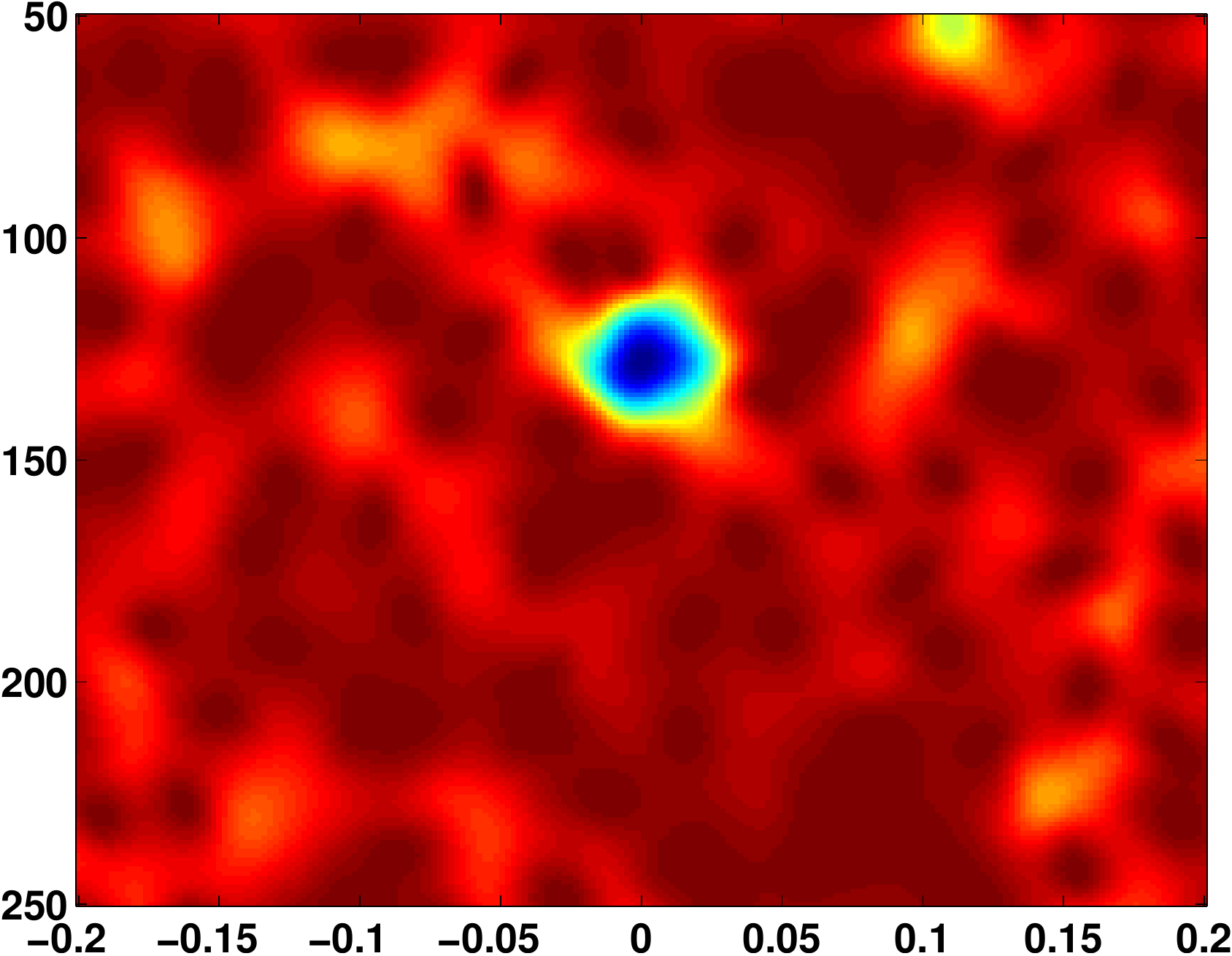} \\
		(c) & (d) & (e)
	\end{tabular}
\caption{\small\sl (a) The original signal $\vh_0$ is a raised cosine window on the interval $|t|\leq 5/128$ and is modulated to frequency $\omega_0 = 2\pi\cdot 128$.  (b)  The energy difference image $\|\vh_0-\Pth\vh_0\|_2^2$.  Each point in this image is the difference between the energy in $\vh_0$ and the energy of its orthogonal projection onto a two dimensional subspace corresponding to a particular shift and modulation frequency, as in \eqref{eq:gaborct}--\eqref{eq:gaborSth}.  (c) One realization of the compressed energy difference image $\|\vy-\tilde\Pth\vy\|_2^2$ formed from $\vy=\mPhi\vh_0$ with $M=10$; (d) a realization with $M=20$; (e) $M=30$.  The minima in (c),(d), and (e) are all in exactly the same location as in (b), but clearly gets more distinct as $M$ increases.
}
\label{fig:cGabor}
\end{figure}


\subsubsection{Fast forward modeling for source localization}

Another application of compressed subspace matching is accelerating modeling algorithms for source localization.  For example, \cite{jasa2012compressive} shows how a {\em matched-field processing} (MFP) algorithm, working directly in the compressed domain, can be used to locate an underwater acoustic target in range and depth based on recorded frequency-domain acoustic measurements at a series of receivers at known locations.  Although the spectral content of the source is generally unknown, there exists a model for the frequency response, the {\em Green's function} $\mG_\theta(\omega) \in \C^N$, between any hypothetical source location $\theta = (r,z)$ (containing range and depth) and receivers $n=1,\ldots,N$.  If a source at frequency $\omega$ and complex amplitude $\beta(\omega) \in \C$ is active at location $\theta$, the response at the $N$ receivers is modeled as
\[
	\vh(\omega) = \beta(\omega)\mG_\theta(\omega) + \mathrm{noise}.
\]
Given observations at one or more frequencies\footnote{If the source signature is known across frequency, the responses $\mG_\theta(\omega)$ are added coherently into a vector of length $N$; if the signature is unknown, the responses are concatenated into a longer vector.}, we can estimate the source location by solving \eqref{eq:sm} where the (real-valued, two-dimensional) subspaces are given by the column space of the $2N\times 2$ matrices
\[
	\mG'_\theta = 
	\begin{bmatrix}
		\Re{\mG_\theta} & -\Im{\mG_\theta} \\ \Im{\mG_\theta} & \Re{\mG_\theta}
	\end{bmatrix}
\]

Typically, the responses $\mG_\theta$ are only known implicitly, coming from higher-level knowledge about the geometry and materials in the channel.  They can be acquired through simulating a source at each receiver location, and backpropagating the response over the locations of interest; if there are many receivers and the simulations are done across many frequencies, these simulations can be computationally expensive.  In \cite{jasa2012compressive}, it was demonstrated that the matched field processing could be done effectively from a compressed version of the $\mG_\theta$.  Sources with randomly generated amplitudes were simulated at all receiver locations simultaneously, resulting in a series of indirect observations $\vphi_1^\T\mG'_\theta,\ldots,\vphi_M^\T\mG'_\theta$ of the Green's function.  These computations can be stacked on top of one another to form $\mPhi\mG'_\theta$ for all $\theta$ --- then a set of given observations can be encoded as $\vy=\mPhi\vh$ and the compressed version \eqref{eq:csm} of the problem is solved.

Figure~\ref{fig:cMFP} shows the classical and compressive error surfaces for a shallow-water channel as a function of range and depth using $M=6$ compressive measurements. Here, the Green's function is modeled as a Pekeris waveguide \cite{jensen1994computational} at a frequency of $150~\text{Hz}$ with $N=37$ uniformly spaced receivers along a vertical array between $10$ meters and $190$ meters in depth. Note that although the compressive error surface is a rough approximation to the true error surface, they are close enough to share a characteristically similar minimizer.

\begin{figure}
	\centering
	\begin{tabular}{cc}
		\includegraphics[width=3in]{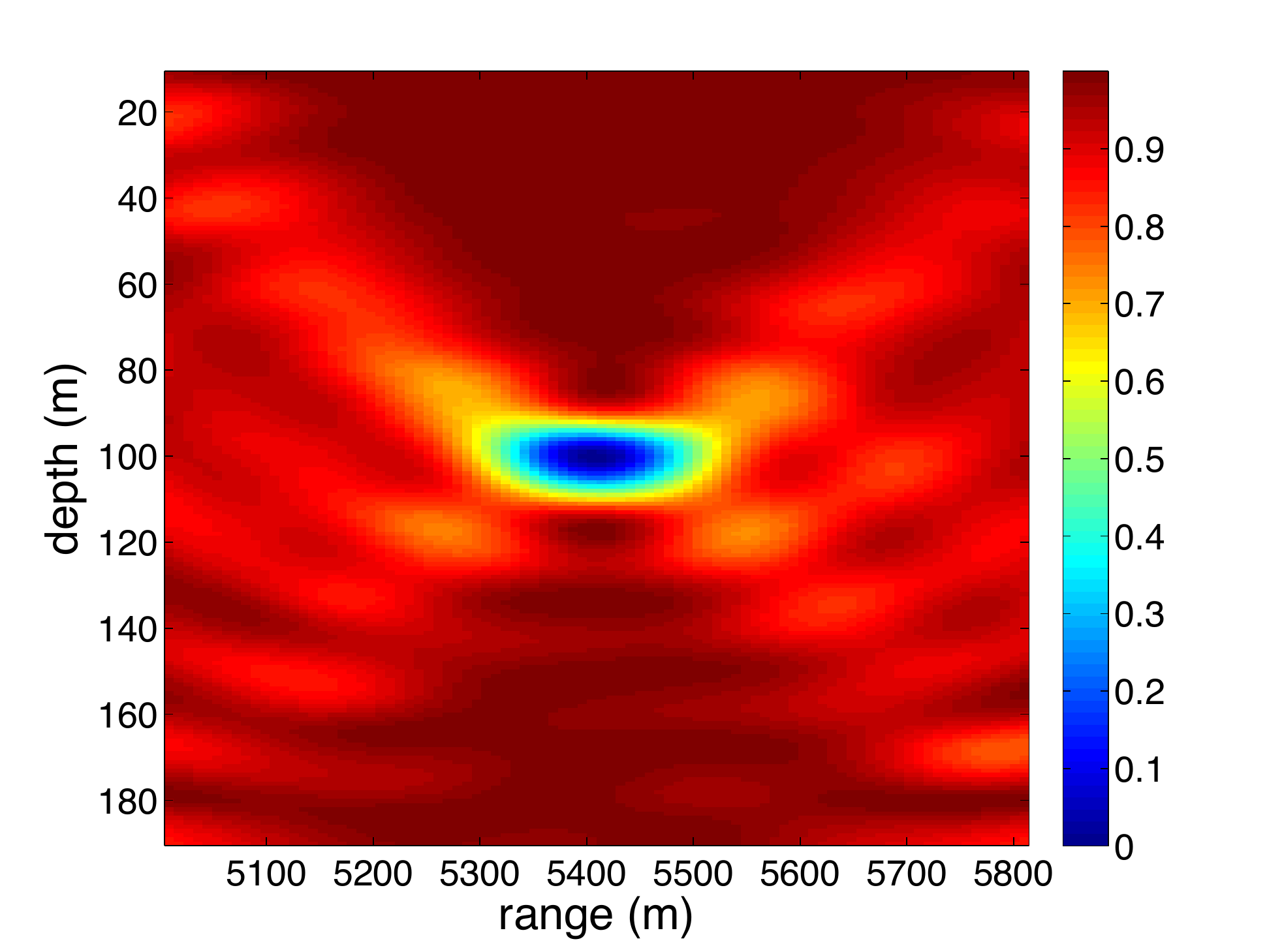} &
		\includegraphics[width=3in]{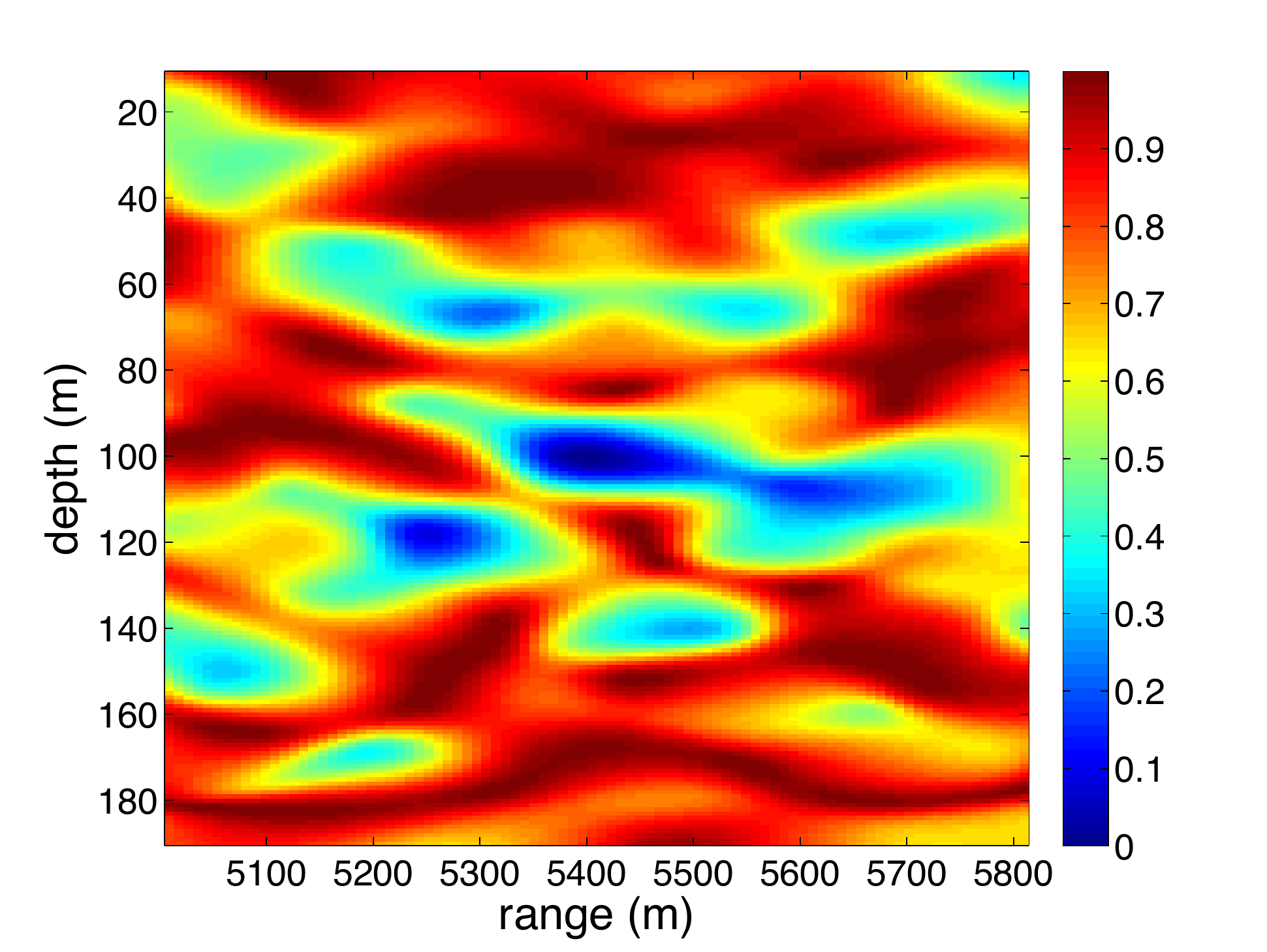} \\
		(a) & (b)
	\end{tabular}
	\caption{\small\sl A subspace matching problem for an underwater acoustic channel.  (a) The energy difference $\|\vh_0-\Pth\vh_0\|_2^2$ for source located at $\theta=(r,z)=(5405{\rm m},100 {\rm m})$.  This image was formed using full knowledge of the (real-valued) Green's function $\mG'_\theta(\omega)$ at single frequency.  (b) The difference in the compressed domain formed using $\mPhi\mG'_\theta$ with $M=6$.}
	\label{fig:cMFP}
\end{figure}


\subsection{Linear encoding of a collection of subspaces}
\label{sec:linear-encoding}


A key part of analyzing the direct-versus-compressed performance gap in \eqref{eq:Ediff} is showing that under appropriate conditions on the number of rows, the random $M\times N$ matrix $\mPhi$ provides a {\em stable embedding} for the collection of subspaces $\{\calS_\theta\}$.  Theorem~\ref{th:ProjectEmbed} below says that if $\vx_1$ and $\vx_2$ are both in (possibly different) subspaces in $\{\calS_\theta\}$, say $\vx_1\in\calS_{\theta_1},~\vx_2\in\calS_{\theta_2}$, then with high probability, the linear mapping $\mPhi$ preserves their distance.  That is, there exists a $\delta > 0$ such that
\begin{equation}
	\label{eq:subspacerip}
	(1-\delta)\|\vx_1-\vx_2\|_2^2 ~\leq~
	\|\mPhi(\vx_1-\vx_2)\|_2^2 ~\leq~ 
	(1+\delta)\|\vx_1-\vx_2\|_2^2
	\quad \vx_1\in\calS_{\theta_1},~\vx_2\in\calS_{\theta_2},~\text{for all }\theta_1,\theta_2\in\Theta.
\end{equation}
In fact, Theorem~\ref{th:ProjectEmbed} shows that we can take
\begin{equation}
	\label{eq:deltabound1}
	\delta ~\lesssim~ \sqrt{\frac{K(\Delta+\log K)}{M}},
\end{equation}
where $\Delta$ is the same measure of geometric complexity of $\{\calS_\theta\}$ as in the subspace matching bound in \eqref{eq:Ebound1}.  This means that the embedding is uniformly stable for all $\theta_1,\theta_2$ when the number of rows in $\mPhi$ is just a little larger than the dimension $K$ of each subspace.

Random projections are known to be efficient stable linear embeddings in sparse signal models, manifolds, etc., so it instructive to compare the content of \eqref{eq:subspacerip} and \eqref{eq:deltabound1} to previous results in the literature.  For an arbitrary $\calQ\subset\R^N$, let $\delta(\calQ)$ be defined as
\[
	\delta(\calQ) = \max_{\vx_1,\vx_2\in\calQ} \frac{\left|\|\mPhi(\vx_1-\vx_2)\|_2^2 - \|\vx_1-\vx_2\|_2^2\right|}{\|\vx_1-\vx_2\|_2^2}.
\]
When $\calQ$ is a finite set of size $|\calQ|=Q$, then a now classical result knows as the {\em Johnson-Lindenstrauss} Lemma \cite{johnson84ex,dasgupta03el} tells us that we can take 
\[
	\delta(\calQ) \lesssim \sqrt{\frac{2\log Q}{M}},
\]
meaning that $\mPhi$ preserves distances when $M\gtrsim\log Q$.  

An analog to the Johnson-Lindenstrauss Lemma has been developed for arbitrary sets $\calQ$ \cite{klartag05em,mendelson07re,dirksen13ta} (i.e.\ sets with infinite cardinality).  In these works, the $\log Q$ above is (approximately speaking) replaced with the {\em Gaussian width}.  These results are closely related to our Theorem~\ref{th:ProjectEmbed}, with the main difference being that our bounds explicitly rely on the geometry of the set of subspaces, rather than a set of points.

When $\calQ$ is a linear subspace of dimension $K$, then $\delta(\calQ)$ is directly related to the singular values of $\mPhi$; in particular, it is the maximum deviation of the singular spectrum from $1$.  Random matrices of this kind are well conditioned when they are overdetermined by a constant factor \cite{marchenko67di,davidson01lo,vershynin12in}; we can take  
\[
	\delta(\calQ) \lesssim \sqrt{\frac{K}{M}}.
\]

A straightforward extension \cite{baraniuk08si,candes06ne,vershynin12in} of this result yields a uniform embedding guarantee when $\calQ$ is all vectors in a finite collection of subspaces $\{\calS_\theta~:~\theta\in\Theta\}$.  In this case, 
\begin{equation}
	\label{eq:csdelta}
	\delta(\calQ) \lesssim \sqrt{\frac{2K + 2\log|\Theta|}{M}},
\end{equation}	
where $|\Theta|$ is the number of subspaces.  In the literature on compressive sensing, where the goal is to recover a sparse vector from linear measurements through $\mPhi$, stable recovery for all $K$-sparse vectors can be guaranteed \cite{candes06st,candes08re} by considering $\calQ$ to be the set of all vectors in $\R^N$ with at most $K$ active components, a union of subspaces of size $|\Theta| = {N\choose K} \lesssim e^{K\log(N/K)}$.  Sparse vectors, then, can be encoded using $M\gtrsim K\log(N/K)$ random measurements.

In this paper, we are interested in collections of subspaces indexed by $\theta$ on the continuum, meaning that $|\Theta|$ is (uncountably) infinite.  This requires a more nuanced treatment than simply applying \eqref{eq:csdelta}.  The bound in \eqref{eq:csdelta} is derived by treating each subspace individually, and so depends only on the number of subspaces to treat.  We will more carefully take advantage of the relationships between subspaces that are very close to one another.  The geometric constant $\Delta$ in \eqref{eq:deltabound1} is, roughly speaking, a measure of how well the collection $\{\calS_\theta\}$ can be approximated by a finite cover across many resolutions.

There is another important result in the sparse and low-rank recovery literature that explicitly handles a continuum of subspaces.  In \cite{recht10gu}, a uniform embedding result was presented for the set of $N\times N$ matrices of rank at most $K$ (this result was refined in \cite{candes11ti}) passing through a random projection (in this case, $\mPhi$ would be $M\times N^2$).  The argument essentially reduces to taking $\calQ$ to be {\em all} (pairs of) $K$ dimensional subspaces of $\R^N$, with 
\[
	\delta(\calQ) \lesssim \sqrt{\frac{NK}{M}}.
\]
In this paper, we consider (a possibly very small) subset of these spaces, and so in general do not require $M$ to be as large.


\subsection{Related work} 


The compressed subspace matching problem, and its analysis put forth in this paper, is closely related to {\em compressed sensing}, where the goal is to recover a structured (i.e.\ sparse) vector from an underdetermined (and usually random) linear observations.  As mentioned in Section~\ref{sec:linear-encoding}, many of the results in this field hinge on a measurement matrix $\mPhi$ stably embedding a collection of subspaces \cite{candes06ne,donoho06co,candes06st}.  Embedding results have been derived for different finite collections of subspaces of interest to different applications, including sparse signals \cite{baraniuk08si}, signals with certain types of ``group sparsity'' \cite{eldar10bl,boufounos11sp}, and arbitrary finite-size collections \cite{blumensath09sa,eldar09ro}.  In contrast to these works, this paper focusses on a union of subspaces indexed by a continuous parameter.  We could of course sample $\theta$ on a grid and apply these known results, but the bounds acquired this way would depend on the size of the grid chosen (which is deeply unsatisfying in that it violates the spirit of the continuum), and there are invariably complications with values of $\theta$ that are ``off the grid'' (between grid points).  Our approach, taking $\theta$ as a point in the continuum, is simply more direct.  The recent work \cite{tang13ne}, which looks at the particular case of estimating the frequencies of sinusoids, works from this same perspective.

The recent paper \cite{blumensath11sa} also discusses the general problem of recovering a signal which belongs to an arbitrary union of subspaces.  This work essentially shows that if $\mPhi$ is an embedding over this space, then the signal can be reconstructed using a projected Landweber algorithm, where each iteration includes a solution of \eqref{eq:sm}.  This paper provides guarantees for the direct solution of \eqref{eq:csm} --- on the continuum, \eqref{eq:sm} and \eqref{eq:csm} are comparable computationally --- and provides conditions under which $\mPhi$ is likely to provide such an embedding.

As we discuss in Section~\ref{sec:motivating-applications}, a scenario in which solving the subspace matching problem is of particular interest is in estimating the time-of-arrival of a pulse.  Specialized algorithms for this particular problem have also been developed in \cite{gedalyahu10ti} and \cite{yoo12co}.  More general algorithms for recovering signals from unions of subspaces can be found in \cite{ekanadham11re,fyhn13co}.

Other closely related recent results come from the area of manifold embeddings.  In \cite{baraniuk09ra}, it is shown that pairwise distances between points on a manifold are preserved through a random projection; they derive a sufficient condition on the dimension of this projection which involves the volume of the manifold, its ambient and intrinsic dimensions, and its curvature.  This result was later refined in \cite{clarkson08ti}, which removed the dependence on the ambient dimension and used a more refined definition of curvature.  Variations on these results for a variety of different projection operators can be found in \cite{yap13st}.  In this work, we are explicitly concerned with collections of subspaces that are smoothly parameterized rather than a collection of vectors.  One advantage of this is it allows for a more natural treatment of scale --- we do not have any size constraints on our signals of interest.  Similar to these manifold embedding results, our bounds do depend on quantifying the geometry of the set of interest.  The main geometrical concept we use, the size of an finite approximation of the collection of subspaces to a certain precision, is different than the descriptors typically used to manifolds (curvature, volume, etc).

One of the keys to establishing that the gap in \eqref{eq:Ediff} is small is showing that $\mPhi$ preserves the distances between the subspaces we are considering; this is the content of Theorem~\ref{th:ProjectEmbed} in Section~\ref{sec:main-results}.  As mentioned in Section~\ref{sec:linear-encoding}, there are closely related results in the probability literature.  In \cite{klartag05em}, a type of Johnson-Lindenstrauss lemma for arbitrary sets was put forth where the geometry of the set was measured  using the ``generic chaining complexity''.  The methods developed there were refined in \cite{mendelson10em,mendelson07re}, and an explicit relationship between the dimension of the projection and the accuracy of the embedding was proven in \cite{dirksen13ta}.  In \cite{bourgain13to}, similar embedding results are established for sparse random matrices.  In this work, we explicitly quantify the geometrical complexity of the collection of subspaces in terms of the natural metric (operator norm of the distance between projectors) which allows us to calculate the bounds explicitly for several scenarios of immediate interest in signal processing (see Section~\ref{sec:geometric-regularity}).

Very recent and independent work \cite{dirksen14di} applies advanced chaining techniques directly to the problem of embedding a collection of subspaces through a random matrix.  It is possible that the results in this work could remove the logarithmic factor in Theorem~\ref{th:ProjectEmbed}.

\section{Main results}
\label{sec:main-results}


For a fixed $\vh_0$, we wish to bound the difference in the approximation error $\hat{E}^2$ of the subspace matched from indirect measurements $\vy = \mPhi\vh_0$ and the approximation error $\bar{E}^2$ of the subspace matched from a direct observation.  (The direct observation error can be thought of as an oracle error).  Our analysis will assume that the matrix $\mPhi$ is Gaussian, with
\begin{equation}
	\label{eq:PhiGauss}
	\Phi[m,n] ~\sim~ \mathrm{Normal}\left(1,M^{-1}\right),
\end{equation}
and all of the entries being independent.

To do this, we will establish a {\em uniform bound} on the difference between the energy of the projection of $\vh_0$ onto subspace $\calS_\theta$ is close to the energy of the projection of $\vy$ onto the range of $\mPhi\mP_\theta$.  We will show that
\begin{equation}
	\label{eq:supW}
	\sup_{\theta\in\Theta} \left|\, \|\tilde\Pth\vy\|_2^2 - \|\Pth\vh_0\|_2^2 \right| ~\leq~\delta,
\end{equation}
where $\delta$ depends on the number of rows $M$ and the geometry of the collection of subspaces $\{\calS_\theta\}$.  It immediately follows that
\begin{align*}
	\hat{E}^2 - \bar{E}^2 &= \|\mP_{\bar\theta}\vh_0\|_2^2 - \|\mP_{\hat\theta}\vh_0\|_2^2 \\
	&\leq \|\tilde\mP_{\hat\theta}\vy\|_2^2 - \|\mP_{\hat\theta}\vh_0\|_2^2 + 
	\|\mP_{\bar\theta}\vh_0\|_2^2 - \|\tilde\mP_{\bar\theta}\vy\|_2^2 \\
	&\leq 2\delta,
\end{align*}
where the first inequality follows from the fact that $\|\tilde\mP_{\hat\theta}\vy\|_2^2 > \|\tilde\mP_{\bar\theta}\vy\|_2^2$ by the definition of $\hat\theta$ in \eqref{eq:csm},
\[
	\hat\theta = \arg\min_{\theta\in\Theta}\|\vy-\tilde\Pth\vy\|_2^2 = \arg\max_{\theta\in\Theta}\|\tilde\Pth\vy\|_2^2.
\]

With $\mPhi$ random, we can interpret \eqref{eq:supW} as a bound on the supremum of a {\em random process} indexed by $\theta\in\Theta$.  Our upper bound $\delta$ will not hold in an absolute sense, but rather with an appropriately high probability.  The supremum in \eqref{eq:supW} follows from bounding the suprema of two related random processes:
\begin{align*}
	\text{(C1)}\qquad&\sup_{\theta\in\Theta}\|\Pth - \Pth\mPhi^\T\mPhi\Pth\| \leq \delta_1, \\
	\text{(C2)}\qquad&\sup_{\theta\in\Theta}\|\Pth\mPhi^\T\mPhi\Pth^\perp\vh_0\|_2 \leq \delta_2.
\end{align*}
The first bound (C1) tells us that all of the projection operators $\Pth$ retain their structure when viewed through the matrix $\mPhi$.  The norm in (C1) is the standard operator norm (largest magnitude eigenvalue in this case).  The second bound (C2) tells us that $\Pth^\perp\vh_0$ stays approximately orthogonal to the subspace $\calS_\theta$ after being passed through $\mPhi^\T\mPhi$.

Our first main result, which we prove in Section~\ref{sec:supWproof} relates (C1) and (C2) to our main goal \eqref{eq:supW}.
\begin{theorem}
	\label{th:supW}
	For a matrix $\mPhi$ and collection of subspaces $\{\calS_\theta,~\theta\in\Theta\}$ as described above, if (C1) and (C2) hold with $\delta_1,\delta_2<1$, then
	\[
		\sup_{\theta\in\Theta} \left|\, \|\tilde\Pth\vy\|_2^2 - \|\Pth\vh_0\|_2^2 \right| ~\leq~\delta := 
		\frac{3\delta_1+2\delta_2 + (\delta_1+\delta_2)^2}{1-\delta_1}.
	\]
\end{theorem}

The upper bounds in (C1) and (C2) depend on the {\em geometrical structure} of the collection of subspaces $\{\calS_\theta\}$.  This structure is quantified through how well $\{\calS_\theta\}$ can be approximated at different resolutions using a finite subset of the subspaces.  This approximation is taken using the natural distance between two subspaces, the  operator norm of the difference of their projectors:
\[
	d(\calS_\tha,\calS_\thb) \triangleq \|\Ptha-\Pthb\|.
\]
Notice that it is always true that $0\leq d(\calS_\tha,\calS_\thb)\leq 1$.  A finite subset $\{\calT_q := \calS_{\theta_q},~~\{\theta_1,\ldots,\theta_Q\}\subset\Theta\}$ is an $\epsilon$-approximation to $\{\calS_\theta\}$ if everything in $\{\calS_\theta\}$ is within distance $\epsilon$ of a corresponding element in $\{\calT_q\}$,
\[
	\sup_{\theta\in\Theta}\min_{1\leq q\leq Q} d(\calS_\theta,\calT_q)
	~\leq~ \epsilon.
\]
The covering number of $\{\calS_\theta\}$ at resolution $\epsilon$ it the size of the smallest such $\epsilon$-approximation
\[
	N(\{\calS_\theta\},\epsilon) = 
	\min\left\{\operatorname{card}(\{\calT_q\}) ~:~ 
	\sup_{\theta\in\Theta}\min_{q} d(\calS_\theta,\calT_q)
	~\leq~ \epsilon\right\}.
\]
Of course, as $\epsilon\rightarrow 0$, the covering number $N(\{\calS_\theta\},\epsilon)$ increases.  Our assumption about the collection of subspaces $\{\calS_\theta\}$ is that this covering number grows at most polynomially as $\epsilon$ decreases.  That is, there exist constants $N_0$ and $\alpha$ such that
\begin{equation}
	\label{eq:georeg}
	N(\{\calS_\theta\},\epsilon)~\leq~ N_0\,\epsilon^{-\alpha},
	\quad\text{for all}\quad 0\leq\epsilon\leq 1.
\end{equation}
We call such collections of subspaces {\em geometrically regular}.  As it turns out, this assumption is met for many and perhaps even most cases of practical interest.  In Section~\ref{sec:geometric-regularity}, we compute bounds on $(N_0,\alpha)$ in a few cases where $\Theta$ is a subset of the real line that indexes the position (shift) of $K$ orthogonal basis functions.  In each of these cases, we see that $N_0\leq\operatorname{poly}(K)\cdot|\Theta|$ and $\alpha=1$ or $2$.  In general, we will treat $\alpha$ like a small constant, with $N_0$ capturing the most important parts of the geometrical complexity of the set $\{\calS_\theta\}$.

In the bounds below for (C1) and (C2), $(N_0,\alpha)$ are combined into a single constant $\Delta$ which has the form
\begin{equation}
	\label{eq:Delta}
	\Delta = \log(8^\alpha N_0^2) + 2 ~\approx~ 2.08\,\alpha+ 2\log N_0 + 2.
\end{equation}
Roughly speaking, this quantity approximates (the logarithm of) the number of subspaces over which the processes $\|\Pth - \Pth\mPhi^\T\mPhi\Pth\|$ and $\|\Pth\mPhi^\T\mPhi\Pth^\perp\vh_0\|_2$ are significantly different.

Our next main result relates gives us a means to control the upper bound $\delta_1$ in (C1) in terms of these geometrical descriptors of $\{\calS_\theta\}$.  In both Theorems~\ref{th:ProjectEmbed} and \ref{th:suphperp}, we assume that $M\geq K$, which is a necessary condition to stably embed even a single subspace of dimension $K$.
\begin{theorem}
	\label{th:ProjectEmbed}
	Let $\mPhi$ be an $M\times N$ matrix with independent Gaussian entries as in \eqref{eq:PhiGauss}, and let $\{\calS_\theta\}$ be a collection of subspaces that are geometrically regular in the sense of \eqref{eq:georeg} with $\Delta$ defined as in \eqref{eq:Delta}.  Then there exist constants $C_1$ and $C_2$ such that for all $t\geq 0$ 
	\begin{equation}
		\label{eq:ProbEmbedFail}
		\P{\sup_{\theta\in\Theta}\|\Pth - \Pth\mPhi^\T\mPhi\Pth\| > C_1\,\gamma(t)}
		~\leq~
		C_2K\,\e^{-t+\Delta},
	\end{equation}
	where
	\begin{equation}
		\label{eq:gammat}
		\gamma(t) = \sqrt{\frac{Kt}{M}} + \frac{Kt}{M} + \frac{t}{\sqrt{M}} + 
		\frac{t^2}{M} + \sqrt{\frac{K}{M}} + \frac{\alpha^2}{M}.
	\end{equation}
\end{theorem}
To make the probability bound in \eqref{eq:ProbEmbedFail} meaningful, we will need to take $t\gtrsim \Delta+\log K$. It follows that, with high probability, we have the following upper bound on $\delta_1$:
\[
	\delta_1 ~\lesssim~ \left( \frac{\max\left(K(\Delta+\log K), (\Delta+\log K)^2\right)}{M} \right)^{1/2}.
\]
This bound will be meaningful when we have:
\[
	M~\gtrsim~ \max\left(K(\Delta+\log K), (\Delta+\log K)^2\right).
\]
As we will see in the examples in Section~\ref{sec:geometric-regularity}, it is typical that $\Delta$ scales like $\log K$ as well.

This $\delta_1$ term also gives an immediate bound on the embedding all signals from these subspaces with respect to the random projection $\mPhi$. To wit, for any $\delta_1$ satisfying (C1), we have:
\begin{equation}
\sup_{\theta \in \Theta} \sup_{f \in \calS_{\theta}} \left| \frac{\|\mPhi f\|^2 }{\| f \|^2} - 1 \right| 
 =  \sup_{\theta \in \Theta} \sup_{\|f\| = 1} \| f^T \Pth f - f^T \Pth\mPhi^\T\mPhi\Pth f \| = 
\sup_{\theta \in \Theta} \|\Pth - \Pth\mPhi^\T\mPhi\Pth\| \leq \delta_1,
\end{equation}
mirroring the classic expression for the restricted isometry property of sparse vectors.

Our final main result is a tail bound of a similar form for the upper bound $\delta_2$ in (C2).
\begin{theorem}
	\label{th:suphperp}
	With $\mPhi$ and $\{\calS_\theta\}$ as in Theorem~\ref{th:ProjectEmbed}, there exist constants $C_1$  and $C_2$ such that for any fixed $\vh_0\in\R^N$ with $\|\vh_0\|_2=1$,
	\[
		\P{\sup_{\theta\in\Theta} \|\Pth\mPhi^\T\mPhi\Pth^\perp\vh_0\| > C_1\,\gamma(t)}
		~\leq~
		C_2K\e^{-t+\Delta}
	\]
	where 
	\[
		\gamma(t) = \sqrt{\frac{tK}{M}} + \frac{t\sqrt{K}}{M} + \sqrt{\frac{K}{M}}.
	\]
\end{theorem}
As with Theorem~\ref{th:ProjectEmbed}, we can make the probability above small by taking $t\gtrsim\Delta+\log K$, meaning that $\sup_{\theta\in\Theta} \|\Pth\mPhi^\T\mPhi\Pth^\perp\vh_0\|\leq\delta<1$ for $M\gtrsim \delta^{-2}K(\Delta+\log K)$.

\section{Collections of subspaces with geometric regularity}
\label{sec:geometric-regularity}


In this section, we will discuss some specific examples of a parameterized collection of subspaces $\{\calS_{\theta}\}$ that are geometrically regular in that there exist constants $N_0,\alpha$ such that
\begin{equation*}
	N( \{\calS_{\theta}\} ,\epsilon, d) \leq N_0 \epsilon^{-\alpha},
	\quad d(\calS_{\theta_1},\calS_{\theta_2}) = \|\Ptha-\Pthb\| \quad \forall \epsilon \leq 1.
\end{equation*}
In each of the examples, mapping from $\theta\in\R^D$ to the $K$-dimensional subspace $\calS_\theta$ is Holder-continuous.  For $D=1$, this means that
\begin{equation}
	\label{eq:holder}
	\|\mP_{\theta_1} - \mP_{\theta_2} \| \leq \beta|\theta_1 - \theta_2|^{\rho},
\end{equation}
for some constants $\beta$ and $\rho$.  This allows us to directly tie the covering numbers of the set of subspaces $\{\calS_\theta\}$ to the covering numbers of the parameter class $\Theta$ under the standard Euclidean metic as
\begin{align}
	\label{eq:NSNT}
	\|\mP_{\theta_1} - \mP_{\theta_2} \| \leq \beta |\theta_1 - \theta_2|^{\rho}
	& ~\Rightarrow~ 
	N(\{ \calS_{\theta}\},\epsilon,d) \leq N(\Theta,(\epsilon/\beta)^{1/\rho},|\cdot|). 
\end{align}
The sets $\Theta$ we consider are simple: in two of the three examples below, they are just intervals on the real line.  In this case, we have the easy bound
\[
	N([a,b], \epsilon, |\cdot|)\leq \frac{|b-a|}{2\epsilon}+1 \leq \frac{|b-a|}{\epsilon},\quad \forall \epsilon \leq (1 / \beta)^{1/\rho},
\]
provided that $|b-a| \geq (\epsilon/\rho)^{1 / \rho}$. So if, for example, we establish \eqref{eq:holder} and $\Theta = [a,b]$, we can take $\alpha = 1/\rho$ and $N_0 = \beta^{1/\rho}|b-a|$.

%

All of our examples involve collections of subspaces index by a shift parameter or (equivalently) a modulation parameter.  The general idea is that the signal lies in the span (or can be closely approximated by) a superposition of $K$ basis functions that are localized in time; the parameter $\theta$ indicates the location of the basis functions.  More precisely, for the shift parameter spaces, we are interested in signals that can be written
\[
	h(t) = \sum_{k=1}^K h[k] v_k(t-\theta),
	\quad\text{for some $\vh\in\R^K$ and $\theta\in[t_1,t_2]$,}
\]
where $v_1(t),\ldots,v_K(t)\in L_2(\R)$ are linearly independent, continuous-time signals.  The collection of all such signals form a manifold in $L_2(\R)$.  Our model will be that we observe this signal (through $\Phi$) after it has been projected onto a fixed $N$ dimensional ambient space $\calH$.  The subspaces $\calS_{\theta}\subset\calH$ are then the span of the $v_k(t-\theta)$ projected onto $\calH$.  Note that we are {\bf not} discretizing the parameter $\theta$; we are simply articulating the signal model as a subset of $\R^N$ rather than $L_2(\R)$.  The exact nature of this discretization is not that critical for our purposes here; we simply introduce it to avoid problematic topological issues that are ultimately irrelevant to our understanding.  The ambient dimension $N$ does not appear anywhere in our results, and so this discretization can be arbitrarily fine-grained. 
We will perform our calculations below based on the $L_2(\R)$ distance between continuous-time signals, and then use the simple fact that $\|\mP_{\calH}\left[u(t)-w(t)\right]\|_2\leq\|u(t)-w(t)\|_{L_2(\R)}$.

In each of our examples below, the $\{v_k(t-\theta)\}$ are a well-defined orthobasis for each subspace in the collection (before it is projected into $\calH$).  We bound the distance between the subspaces parameterized by $\theta_1$ and $\theta_2$ by looking at the distances between these basis functions at different shifts.  If we set $u_1(t) = v_1(t-\theta_1),\ldots,u_{K}(t)=v_K(t-\theta_1)$  and $w_1(t) = v_1(t-\theta_2),\ldots,w_K(t) = v_K(t-\theta_2)$, then we have
\begin{equation}
	\label{eq:Pvbound}
	\|\Ptha-\Pthb\| \leq 2\sup_{\|x(t)\|_2=1}\left(\sum_{k=1}^K|\<x(t),u_k(t)-w_k(t)\>|^2\right)^{1/2}
	\leq 2\left(\sum_{k=1}^K\|u_k(t)-w_k(t)\|_2^2\right)^{1/2}.
\end{equation}


The examples here are meant to be illustrative more than precise.  The geometrical constant $\Delta$ in Theorems~\ref{th:ProjectEmbed} and \ref{th:suphperp} scale like $\Delta\sim \alpha + \log N_0$, and so our main concern will be establishing $\alpha$ as a small constant, and $N_0$ a polynomial of $K$ of modest order.


\subsection{Pulse at an unknown delay}
\label{sec:pulsedelay}

In our first example, we consider the problem of estimating the shift (or ``time of arrival'') of a single pulse with a known shape.  In this case, solving the direct observation problem \eqref{eq:sm} is equivalent to the classical {\em matched filter}; the indirect observation problem in \eqref{eq:csm} has been called the {\em smashed filter} in the literature \cite{davenport07sm,davenport10si} and has been previously analyzed for specialized $\mPhi$ \cite{eftekhari13ma}.

We start by considering a Gaussian pulse,
\[
	v(t) = \pi^{-1/4}\sigma^{-1/2}e^{-t^2/2\sigma^2},
\]
where $\sigma$ is a fixed width parameter and the normalization factor in front ensures $\|v(t)\|_2=1$.  The collection of $K=1$ dimensional subspaces we want to test are generated by the manifold of unit vectors $\{v(t-\theta),~\theta\in[0,T]\}$ projected onto the discretized space $\calH$.  We have
\begin{align}
	\label{eq:gaussautocorr}
	\|v(t-\theta_1)-v(t-\theta_2)\|_2^2 & = 
	\|v(t-\theta_1)\|_2^2 + \|v(t-\theta_2)\|_2^2 - 2\<v(t-\theta_1),v(t-\theta_2)\> \\
	\nonumber
	&= 2\left(1 - e^{-(\theta_2-\theta_1)^2/4\sigma^2}\right) \\
	\nonumber
	&\leq \frac{|\theta_2-\theta_1|^2}{2\sigma^2},
\end{align}
and so applying \eqref{eq:Pvbound},
\begin{align*}
	\|\Ptha-\Pthb\| &\leq \frac{\sqrt{2}}{\sigma}|\theta_2-\theta_1|,
\end{align*}
which combined with \eqref{eq:NSNT} means
\[
	N(\{\calS_\theta\},\epsilon,d) ~\leq~ \frac{\sqrt{2}\,T}{\sigma}\epsilon^{-1}. 
\]
In this case, we can use $\alpha=1$ and $N_0= \sqrt{2}T/\sigma$, and take
\[
	\Delta\leq  \log(T/\sigma) + 4.78,
\]
as the geometrical constant in Theorems~\ref{th:ProjectEmbed} and \ref{th:suphperp}.

As the calculations in \eqref{eq:gaussautocorr} suggest, the distance between two close shifts of the same function is related to the flatness of its autocorrelation function near zero.  For smooth $v(t)$, the following lemma, proved in the appendix, quantifies this behavior in a convenient manner.
\begin{lemma}
	\label{lm:sobolev_shift}
	Let $v(t)$ be a unit-energy signal, $\|v(t)\|_{L_2}=1$, with Fourier transform $\hat{v}(\xi)$.  If $v(t)$ has bounded Sobolev norm,
	\[
		\left(\frac{1}{2\pi}\int_{-\infty}^\infty \xi^2|\hat{v}(\xi)|^2 ~d\xi\right)^{1/2} ~\leq~L,
	\]
	then
	\[
		\|v(t-\theta_1) - v(t-\theta_2)\|_2 ~\leq~ L|\theta_1-\theta_2|.
	\]
\end{lemma}
So if the collection of subspaces $\{\calS_\theta\}$ consists of all shifts $\{v(t-\theta),~\theta\in[0,T]\}$ of a Sobolev differentiable function (projected onto $\calH$), then
\begin{equation}
	\label{eq:sobolevcomplexity}
	N(\{\calS_\theta\},\epsilon,d) \leq LT\,\epsilon^{-1}
	\quad\Rightarrow\quad
	\Delta \leq 2\log (LT) + 4.08.
\end{equation}

If the pulse $v(t)$ is not smooth, we can still bound the distance between shifts using the total variation of $v(t)$.  Roughly speaking, the total variation is the arc length of the curve that $v(t)$ traces out; technically, we define the total variation of a signal on the real line as 
\[
	\|v(t)\|_{TV} = \lim_{a\rightarrow\infty}
	\sup_{\calP}\left\{\sum_{i=1}^{|\calP|}|v(t_{i+1})-v(t_i)|,~\calP = \{t_1,\ldots,t_{|\calP|}\}~\text{is a partition of $[-a,a]$.}\right\}.
\] 
The following lemma relates the total variation the energy in the difference between shifts.
\begin{lemma}
	\label{lm:tv_shift}
	Let $v(t)$ be a signal with finite total variation.  Then
	\[
		\|v(t-\theta_1) - v(t-\theta_2) \|_2 
		~\leq~ 
		\|v(t)\|_{TV}\, |\theta_1-\theta_2|^{1/2}.
	\]
\end{lemma}
As a quick application, suppose that $v(t)$ is a square pulse of length $\sigma$,
\[
	v(t) = 
	\begin{cases}
		1/\sqrt{\sigma}, & -\sigma/2\leq t\leq\sigma/2, \\
		0 & \text{otherwise}.
	\end{cases}
\]
Then $\|v(t)\|_{TV} = 2/\sqrt{\sigma}$, and
\begin{equation}
	\label{eq:squarecomplexity}
	N(\{\calS_\theta\},\epsilon,d) \leq \frac{4T}{\sigma}\epsilon^{-2}
	\quad\Rightarrow\quad
	\Delta \leq 2\log(T/\sigma) + 8.94.
\end{equation}
In cases where $v(t)$ is not smooth, we do expect $\{\calS_\theta\}$ to be richer collection of subspaces, a fact which is reflected in the powers of $\epsilon$ in the bounds \eqref{eq:sobolevcomplexity} and \eqref{eq:squarecomplexity}.  However, the ultimate difference this makes for $\Delta$ is very minor --- it simply results in a greater additive constant.


\subsection{Gabor pulse, unknown shift and unknown modulation frequency}
\label{sec:gabor}

In our next example, we consider the problem of jointly estimating the shift and modulation frequency of a Gabor pulse:
\[
	g(t,\omega;a,b) = 
	e^{-t^2/2\sigma^2}\left(a\cos(\omega t) + b\sin(\omega t)\right).
\]
Each subspace in the collection has $K=2$, and $\calS_{\theta}$, where $\theta = (\tau,\omega)$ is a multiparameter with $D=2$, consists of the projection onto $\calH$ of the span of $e^{-(t-\tau)^2/2\sigma^2}\cos(\omega (t-\tau))$ and $e^{-(t-\tau)^2/2\sigma^2}\sin(\omega (t-\tau))$.  For the search space, we take $\Theta = [0,T]\times [\Omega_l,\Omega_u]$.  Define
\begin{align*}
	v_1(t;\tau,\omega) &= C_1(\omega) e^{-(t-\tau)^2/2\sigma^2}\cos(\omega(t-\tau)),
	\quad C_1(\omega) =  \sqrt{\frac{2}{\sigma\sqrt{\pi}(1+e^{-\omega^2\sigma^2})}}\\
	v_2(t;\tau,\omega) &= C_2(\omega) e^{-(t-\tau)^2/2\sigma^2}\sin(\omega(t-\tau)),
	\quad C_2(\omega) = \sqrt{\frac{2}{\sigma\sqrt{\pi}(1-e^{-\omega^2\sigma^2})}}
\end{align*}
We will assume that the pulse has at least a mild amount of oscillation by restricting $\Omega_l\sigma > 1$.  We will bound the distances in two parts, using
\begin{align*}
	\|v_1(t;\tau_1,\omega_1)-v_1(t;\tau_2,\omega_2)\|_2 &\leq
	\|v_1(t;\tau_1,\omega_1)-v_1(t;\tau_2,\omega_1)\|_2 +
	\|v_1(t;\tau_2,\omega_1)-v_1(t;\tau_2,\omega_2)\|_2,
\end{align*}
and do a similar calculation for $v_2$.

For the time shift, we use Lemma~\ref{lm:sobolev_shift} along with the following expressions for the Fourier transforms:
\begin{align*}
	\hat{v}_1(\xi;\tau,\omega) &= \frac{1}{1+e^{-\omega^2\sigma^2}}
	\left(e^{-(\xi-\omega_1)^2\sigma^2/2} + e^{-(\xi+\omega_1)^2\sigma^2/2}\right)e^{-j\tau\xi} \\
	\hat{v}_2(\xi;\tau,\omega) &= \frac{1}{1-e^{-\omega^2\sigma^2}}
	\left(e^{-(\xi-\omega_1)^2\sigma^2/2} - e^{-(\xi+\omega_1)^2\sigma^2/2}\right)e^{-j(\tau\xi+\pi/2)}.
\end{align*}
This leads to the bound
\begin{align*}
	\|v_1(t;\tau_1,\omega_1)-v_1(t;\tau_2,\omega_1)\|_2&\leq
	\sqrt{\left(\omega_1^2 + \frac{1}{\sigma^2}\right)}\cdot|\tau_1-\tau_2|,
\end{align*} 
and the same bound holds for $\|v_2(t;\tau_1,\omega_1)-v_2(t;\tau_2,\omega_1)\|_2$.  For the frequency shift,
\begin{align*}
	\|v_1(t;\tau_2,\omega_1)-v_1(t;\tau_2,\omega_2)\|_2 &=
	(2\pi)^{-1/2}\|\hat{v}_1(\xi;\tau_2,\omega_1)-\hat{v}_1(\xi;\tau_2,\omega_2)\|_2 \\
	&\leq \max(C_1(\omega_1),C_1(\omega_2))\sigma\left\|e^{-(\xi-\omega_1)^2\sigma^2/2} - e^{-(\xi-\omega_2)^2\sigma^2/2}\right\|_2 \\
	&\leq \sigma|\omega_1-\omega_2|,
\end{align*}
where the last step follows from a similar calculation to the one at the beginning of Section~\ref{sec:pulsedelay}.  The calculation for the frequency shift of the $v_2$ is only slightly different,
\begin{align*}
	\|v_2(t;\tau_2,\omega_1)-v_2(t;\tau_2,\omega_2)\|_2 
	&\leq \max(C_2(\omega_1),C_2(\omega_2))\,\sigma\left\|e^{-(\xi-\omega_1)^2\sigma^2/2} - e^{-(\xi-\omega_2)^2\sigma^2/2}\right\|_2 \\
	&\leq 1.582\,\sigma|\omega_1-\omega_2|,
\end{align*}
where we have used the fact that $\Omega_l\sigma\geq 1$ to bound the maximum in the last inequality.

We can now combine these inequalities to bound the distance between the projectors.  With $\theta_1=(\tau_1,\omega_1)$ and $\theta_2 = (\tau_2,\omega_2)$, we have
\[
	\|\Ptha-\Pthb\| \leq 2\sqrt{2}\left(\left(\Omega_u^2+\frac{1}{\sigma^2}\right)^{1/2}|\tau_1-\tau_2| + 1.582\sigma|\omega_1-\omega_2|\right)
\] 
With $\Omega=\Omega_u-\Omega_l$ as the length of the search interval in frequency and $T$ length of the search interval in time, we can then bound the covering numbers as
\[
	N(\{\calS_\theta),\epsilon,d) ~\leq~ \frac{9\sqrt{\sigma^2\Omega_u^2+1}\,T\,\Omega}{\epsilon^2}
	\quad\Rightarrow\quad
	\Delta\leq2\log((\Omega_u^2\sigma^2+1)^{1/2}T\Omega) + 8.36.
\]

\subsection{Multifrequency pulse, unknown shift}

Consider the problem of trying to find a pulse (again at a location in $[0,T]$) which is localized in time and exists over a known frequency band.  We might model such a signal as a superposition of modulations of a fixed envelope function $g(t)$:
\[
	h(t) = \sum_{k=1}^K \beta_k\, v_k(t),\quad\text{with}~v_k(t) = g(t)\cos(\omega_k t).
\]
With a careful choice of $g(t)$ and frequencies $\omega_k$, the $v_k(t)$ are orthogonal; this is known as the {\em lapped orthogonal transform} \cite{malvar89lo}.  To illustrate, we consider the particular case of 
\begin{displaymath}
	g(t) = \sqrt{\frac{2}{\sigma}}\left\{
    \begin{array}{rrl}
      0 : & t & < -\sigma/4
\\    \sin\left(\frac{\pi}{4}\left(1+\frac{4t}{\sigma}\right)\right) : & -\sigma/4 \leq ~ t & < \sigma/4
\\    1 : & \sigma/4 \leq ~ t & < 3\sigma/4
\\    \sin\left(\frac{\pi}{4}\left(5-\frac{4t}{\sigma}\right)\right) : & 3\sigma/4 \leq ~ t & < 5\sigma/4
\\    0 : & \sigma/4 \leq ~ t &
    \end{array}
  \right. ,
  \quad
  \omega_k = \frac{\pi}{\sigma}\left(k-\frac{1}{2}\right),
\end{displaymath}
for a fixed width $\sigma$.  The $v_k$ are orthogonal, concentrated on $[0,\sigma]$ (and compactly supported on $[-\sigma/4,5\sigma/4]$), have a maximum height of $\sqrt{2/\sigma}$, and have at most $3k/2+2$ extrema.  These facts combine to give us an upper bound on the total variation:
\[
	\|v_k(t)\|_{TV} ~\leq~ \sqrt{2/\sigma}\,(3k+4).
\]
Then Lemma~\ref{lm:tv_shift} and \eqref{eq:Pvbound} combine to give us (assuming $K\geq 2$),
\begin{align*}
	\|\Ptha-\Pthb\| &\leq 2|\theta_1-\theta_2|^{1/2}\left((2/\sigma)\sum_{k=1}^K (3k+4)^2\right)^{1/2} 
	\leq 12.5\cdot K^{3/2}(|\theta_1-\theta_2|/\sigma)^{1/2},
\end{align*}
and so we can take $\alpha=2$ and $N_0=(12.5)^2K^3/\sigma$, and 
\[
	N(\{\calS_\theta\},\epsilon,d)~\leq~ \frac{(12.5)^2K^3T}{\sigma}\epsilon^{-2}
	\quad\Rightarrow\quad
	\Delta \leq 6\log K + 2\log(T/\sigma) + 16.27.
\]
Here we see the dimension $K$ of the subspaces coming out explicitly in the geometric constant.



\section{Technical details}


\subsection{Proof of Theorem~\ref{th:supW}}
\label{sec:supWproof}

Let $\theta$ be any member of $\Theta$, and let $\Vth$ be a $N\times K$ matrix whose columns are an orthobasis for $\calS_\theta$.  We can write $\Pth = \Vth\Vth^\T$ and
\begin{align*}
	\|\Pth - \Pth\mPhi^\T\mPhi\Pth\| = \|\Vth(\mId-\Vth^\T\mPhi^\T\mPhi\Vth)\Vth^\T\| = \|\mId - \Vth^\T\mPhi^\T\mPhi\Vth\|.
\end{align*}
By (C1), all the eigenvalues of $\Vth^\T\mPhi^\T\mPhi\Vth$ are between $1\pm\delta_1$, meaning that $(\Vth^\T\mPhi^\T\mPhi\Vth)^{-1}$ exists and is itself close to the identity.  We can write
\begin{align*}
	\tilde\Pth &= \mPhi\Vth(\Vth^\T\mPhi^\T\mPhi\Vth)^{-1}\Vth^\T\mPhi^\T \\
	&= \mPhi\Vth(\mId + \Hth)\Vth^\T\mPhi^\T,
\end{align*}
where we can bound the size of $\Hth$ using a Neumann series:
\begin{align*}
	\|\Hth\| &= \left\|(\Vth^\T\mPhi^\T\mPhi\Vth)^{-1}-\mId\right\| \\
	&= \left\|\sum_{n=1}^\infty (\mId - \Vth^\T\mPhi^\T\mPhi\Vth)^n\right\| \\
	&\leq \sum_{n=1}^\infty \| \mId - \Vth^\T\mPhi^\T\mPhi\Vth \|^n \\
	&\leq \frac{\delta_1}{1-\delta_1}.
\end{align*}
With $\vy=\mPhi\vh_0$,
\begin{align}
	\nonumber
	\left| \|\tilde\Pth\vy\|_2^2 - \|\Pth\vh_0\|_2^2\right| &=
	\left|\<\vy,\mPhi\Vth(\mId+\Hth)\Vth^\T\mPhi^\T\vy\> - \|\Pth\vh_0\|_2^2\right| \\
	\nonumber
	&\leq \left|\|\Pth\mPhi^\T\vy\|_2^2 - \|\Pth\vh_0\|_2^2\right| + \left|\<\mPhi^\T\vy,\Vth\Hth\Vth^\T\mPhi^\T\vy\>\right| \\
	\nonumber
	&\leq \left|\|\Pth\mPhi^\T\vy\|_2^2 - \|\Pth\vh_0\|_2^2\right| + \|\Vth\Hth\Vth^\T\|\,\|\Pth\mPhi^\T\vy\|_2^2 \\
	\label{eq:Winner}
	&\leq \left|\|\Pth\mPhi^\T\vy\|_2^2 - \|\Pth\vh_0\|_2^2\right| + \frac{\delta_1}{1-\delta_1}\|\Pth\mPhi^\T\vy\|_2^2.
\end{align}
We factor the first term above,
\begin{align*}
	\|\Pth\mPhi^\T\vy\|_2^2 - \|\Pth\vh_0\|_2^2 &=
	\left(\|\Pth\mPhi^\T\mPhi\vh_0\|_2 + \|\Pth\vh_0\|_2\right)\left(\|\Pth\mPhi^\T\mPhi\vh_0\|_2 - \|\Pth\vh_0\|_2\right) 
\end{align*}
and then use the fact that 
\begin{align*}
	\|\Pth\mPhi^\T\mPhi\vh_0\|_2 &\leq \|\Pth\mPhi^\T\mPhi\Pth\vh_0\|_2 + \|\Pth\mPhi^\T\mPhi\Pth^\perp\vh_0\|_2 \\
	&\leq (1+\delta_1)\|\Pth\vh_0\|_2+\delta_2,
\end{align*}
to get
\begin{align*}
	\left|\|\Pth\mPhi^\T\vy\|_2^2 - \|\Pth\vh_0\|_2^2\right| &\leq (2+\delta_1+\delta_2)(\delta_1+\delta_2) =: \delta_3.
\end{align*}
Returning to \eqref{eq:Winner}, we finally have
\begin{align*}
	\left| \|\tilde\Pth\vy\|_2^2 - \|\Pth\vh_0\|_2^2\right| &\leq \delta_3 + \frac{\delta_1(1+\delta_3)}{1-\delta_1} \\
	&= \frac{\delta_1+\delta_3}{1-\delta_1}.
\end{align*}

\subsection{Chaining}
\label{sec:chaining}

The following lemma allows us to easily relate a probabilistic bound on the increments of a random process to a probabilistic bound on its supremum.  We will use it to establish both Theorems~\ref{th:ProjectEmbed} and \ref{th:suphperp}.

\begin{proposition}
	\label{prop:chaining}
	Let $\{\calS_\theta\}$ be a collection of subspaces that is geometrically regular, parameterized by $N_0,\alpha$ as in \eqref{eq:georeg}.  Let $\mZ(\theta)$ be a matrix-valued random process indexed by the same parameters $\theta\in\Theta$.  If for every $\theta_1,\theta_2\in\Theta$ we have 
	\[
		\P{\|\mZ(\theta_1)-\mZ(\theta_2)\| ~\geq~ \gamma(u)\cdot\|\Ptha-\Pthb\|} 
		~\leq~
		\beta\,\e^{-u},
	\]
	where $\beta$ is a constant and $\gamma(u)$ has the form $\gamma(u)=a\sqrt{u} + bu + cu^2$, then for any fixed $\theta_0\in\Theta$,
	\begin{equation}
		\label{eq:supchain}
		\P{\sup_{\theta\in\Theta}\|\mZ(\theta)-\mZ(\theta_0)\| \geq 3\gamma(u) + 6c\log^2(4^\alpha\e)}
		~\leq~
		\beta\,\e^{-u+\Delta},
	\end{equation}
	where $\Delta = \log(8^\alpha N_0^2) + 2$.
\end{proposition}
If $\mZ(\theta)$ is vector-valued, then the above holds with the increment size measured in the standard Euclidean norm, $\|\mZ(\theta_1)-\mZ(\theta_2)\|_2$.

\begin{proof}
	We use a standard chaining argument (see for example \cite[Thm.\ 1.2.7]{talagrand05ge}).
	
	Let $\calT_0 = \{\theta_0\}$, and $\calT_j$ be a series of $\epsilon$-approximations of $\Theta$ with $\epsilon=2^{-j}$.  We know that we can choose these with $\operatorname{card}(\calT_j)\leq N_02^{\alpha j}$.  For any $\theta\in\Theta$, let $\pi_j(\theta)$ be the closest point in $\calT_j$ to $\theta$.  Then for every $\theta$, $\|\mZ(\theta)-\mZ(\theta_0)\|\leq \sum_{j\geq 0}\|\mZ(\pi_{j+1}(\theta))-\mZ(\pi_j(\theta))\|$.  Thus for any given sequence $u_j,~j=0,1,\ldots,\infty$, we have
	\begin{align*}
		\P{\sup_{\theta\in\Theta} \|\mZ(\theta)-\mZ(\theta_0)\| \geq \sum_{j\geq0}\frac{3}{2}2^{-j}\gamma(u_j)}
		&\leq \sum_{j\geq 0}\P{\sup_{\theta\in\Theta}
		\|\mZ(\pi_{j+1}(\theta))-\mZ(\pi_{j}(\theta))\|\geq\frac{3}{2}2^{-j}\gamma(u_j)}.
	\end{align*}
	Now since the number of unique pairs of $\pi_{j+1}(\theta),\pi_j(\theta)$ over all $\theta$ is at most $\operatorname{card}(\calT_j)\operatorname{card}(\calT_{j+1})\leq N_0^22^{\alpha(2j+1)}$, and $\|\mP_{\pi_{j+1}(\theta)} - \mP_{\pi_j(\theta)}\|\leq (3/2)2^{-j}$, we apply the increment tail bound and see that
	\begin{align*}
		\P{\sup_{\theta\in\Theta} \|\mZ(\theta)-\mZ(\theta_0)\| \geq \sum_{j\geq0}\frac{3}{2}2^{-j}\gamma(u_j)}
		&\leq C_a N_0^2\sum_{j\geq 0} 2^{\alpha(2j+1)}\e^{-u_j}.
	\end{align*}
	Taking $u_j = u+ (2\alpha\log 2 + 1)j + (\alpha\log 2 + 2\log N_0 + 1)$ immediately gives us the righthand side of \eqref{eq:supchain}.  The lefthand side comes from applying the following simple lemma, which we prove in the appendix. 
	
	\begin{lemma}
		\label{lm:gsum}
		Let $u_j = u + p j + q$, where $u,p,$ and $q$ are positive constants.  Let
		$\gamma(u) = a\sqrt{u} + bu + cu^2$ with $a,b,c\geq 0$.  Then
		\begin{align}
			\label{eq:gsum}
			\sum_{j\geq 0} \frac{1}{2}2^{-j} \gamma(u_j)
			&\leq
			\gamma(u+p+q) + 2cp^2.
		\end{align}
	\end{lemma}

\end{proof}

\subsection{Proof of Theorem~\ref{th:ProjectEmbed}}
\label{sec:th2proof}

The proof of Theorem~\ref{th:ProjectEmbed} follows from combining Proposition~\ref{prop:chaining}, Lemma~\ref{lm:ProjectIncrement} below, and the fact that for any matrix-valued random process $\mZ(\theta)$,
\begin{align}
	\label{eq:supZtheta}
	\sup_{\theta\in\Theta} \|\mZ(\theta)\| 
	&\leq \sup_{\theta\in\Theta} \|\mZ(\theta)-\mZ(\theta_0)\| + \|\mZ(\theta_0)\|,
\end{align}
for any fixed $\theta_0\in\Theta$.  For a fixed $\theta_0$, we can bound $\|\mZ(\theta_0)\|$ using a standard argument; a bound on the first term above is given by Lemma~\ref{lm:ProjectIncrement} below.

With $\theta_0$ given, the quantity $\|\mP_{\theta_0} - \mP_{\theta_0}\mPhi^\T\mPhi\mP_{\theta_0}\|$ has the exact same distribution as $\|\mId - \mA^\T\mA\|$, where $\mA$ is a $M\times K$ Gaussian random matrix whose entries have a variance of $1/M$.  This random variable describes how far the singular values of an appropriately normalized Gaussian random matrix deviate from $1$, and its behavior is well known.  In particular, we know (see \cite[Lemma 26]{vershynin12in} which is derived from results in \cite{davidson01lo}) that
\[
	\P{\|\mP_{\theta_0} - \mP_{\theta_0}\mPhi^\T\mPhi\mP_{\theta_0}\| > 3\gamma(t)} ~\leq~ 2\e^{-t},
\]
for 
\[
	\gamma(t) = \max(\delta(t),\delta^2(t)),\quad
	\delta(t) = \sqrt{\frac{K}{M}} + \sqrt{\frac{2t}{M}}.
\]

We combine this result with the following lemma, which gives us a bound on $\sup_{\theta\in\Theta} \|\mZ(\theta)-\mZ(\theta_0)\|$ in \eqref{eq:supZtheta}.
\begin{lemma}
	\label{lm:ProjectIncrement}
	Let $\mG_\theta = \Pth - \Pth\mPhi^\T\mPhi\Pth$.  There exists a constant $C$ such that for any fixed $\tha,\thb\in\Theta$,
	\begin{equation}
		\label{eq:projectincrement}
		\P{\|\mG_\tha-\mG_\thb\| \geq C\,\gamma(t)\,\|\Ptha-\Pthb\|}
		~\leq~
		5K\,\e^{-t},
	\end{equation}
	where
	\[
		\gamma(t) =
		\sqrt{\frac{Kt}{M}} + \frac{(K + \log(M/K))t}{M} + \frac{t^2}{M}.
	\]
\end{lemma}

After establishing the lemma, we can choose $\theta_0$ arbitrarily.  Since $K\geq 1$ and the fact that we always have $M\geq\log^2(M)$, we can replace the middle term above with $Kt/M + t/\sqrt{M}$.  Applying the chaining proposition from the previous section adds another $t^2/M$ factor (with a constant which depends on $\alpha$) to $\gamma(t)$ and a factor of $e^\Delta$ to the failure probability.  This establishes the theorem.

\begin{proof}{\bf of Lemma~\ref{lm:ProjectIncrement}}

We will write the increments of $\mG_\theta$ as a sum of independent random matrices, and then use the matrix Bernstein inequality (Proposition~\ref{prop:matrixbernstein}) to establish a tail bound.  Use
$\vphi_m\in\R^N$ to denote the rows of $\mPhi$, and let $\mS = \Ptha+\Pthb$, and $\mD = \Ptha-\Pthb$.
We have 
\begin{align*}
	\mG_\tha-\mG_\thb  
	&= \sum_{m=1}^M\frac{1}{M}\Ptha - \frac{1}{M}\Pthb - \Ptha\vphi_m\vphi_m^\T\Ptha + \Pthb\vphi_m\vphi_m^\T\Pthb \\
	&= \sum_{m=1}^M \frac{1}{M}\mD - \frac{1}{2}\left(\mS\vphi_m\vphi_m^\T\mD + \mD\vphi_m\vphi_m^\T\mS\right) \\
	&=: \sum_{m=1}^M \mX_m ,
\end{align*}
Since $\E[\vphi\vphi^\T]=M^{-1}\mId$, the terms in the sum above are zero mean.

Let $\mP_{\theta_1,\theta_2}$ be the projection operator onto the $2K$ dimensional subspace $\calS_{\theta_1}\cup\calS_{\theta_2}$, and let 
\[
	R_m = \|\mP_{\theta_1,\theta_2}\vphi_m\|_2, \quad \vvarphi_m = \mP_{\theta_1,\theta_2}\vphi_m/R_m.
\]
We will find it useful to condition on the event that none of the $R_m$ are too large,
\[
	\calE(\xi) = \left\{\sqrt{M}\max_{m=1,\ldots,M}R_m \leq\xi\right\}.
\]
Notice that since the $\vphi_m$ are Gaussian, the random scalar $R_m$ and random vector $\vvarphi_m$ are independent; conditioning on $\calE(\xi)$ is the same as conditioning the distribution of each $\mX_m$ on $\calE_m(\xi)=\{\sqrt{M}R_m\leq\xi\}$.

We break \eqref{eq:projectincrement} into two parts:
\begin{equation}
	\label{eq:condGinc}
	\P{\|\mG_\tha-\mG_\thb\| \geq C\,\gamma(t)} ~\leq~
	\P{\|\mG_\tha-\mG_\thb\| \geq C\,\gamma(t)~\mid~\calE_\xi} + \P{\calE(\xi)^c}.
\end{equation}
%
We will bound the first term in \eqref{eq:condGinc} using Proposition~\ref{prop:matrixbernstein}; the second term is subsequently controlled using Proposition~\ref{prop:gausslip} (see \eqref{eq:GincmaxRm} below).


We can rewrite the $\mX_m$ as
\[
	\mX_m = \frac{1}{M}\mD - \frac{R_m^2}{2}\left(\mS\vvarphi_m\vvarphi_m^\T\mD + \mD\vvarphi_m\vvarphi_m^\T\mS\right).
\]
Conditioned on $\calE(\xi)$, an upper bound $B$ in Proposition~\ref{prop:matrixbernstein} is given by
\begin{align}
	\label{eq:GincB}
	\|\mX_m\| &\leq \frac{\|\mD\|}{M} + \frac{\xi^2}{M}\|\mS\vvarphi_m\|_2\|\mD\vvarphi_m\|_2 
	\leq \frac{2\xi^2+1}{M}\|\mD\| =: B,
\end{align}
since by definition $\|\vvarphi\|_2=1$.

For the conditional variance, since $\E[\vvarphi_m\vvarphi_m^\T] = \frac{1}{2K}\mP_{\theta_1,\theta_2}$,
\[
	\E[\mS\vvarphi_m\vvarphi_m^\T\mD + \mD\vvarphi_m\vvarphi_m^\T\mS] = \frac{1}{K}\mD,
\]
and so
\[
	\E[\mX_m^2|\calE_m(\xi)] = \left(\frac{1}{M^2}-\frac{\E[R_m^2|\calE_m(\xi)]}{MK}\right)\mD^2 + 
	\frac{\E[R_m^4|\calE_m(\xi)]}{4}\E\left[\left(\mS\vvarphi_m\vvarphi_m^\T\mD + \mD\vvarphi_m\vvarphi_m^\T\mS\right)^2\right].
\]
We bound this using the fact that $\E[R_m^4|\calE_m(\xi)] \leq \E[R_m^4] = 4K(K+1)/M^2$,  
\begin{align*}
	\left\|\E[\mX_m^2\mid\calE_m(\xi)]\right\| 
	&\leq \frac{\|\mD\|^2}{M^2} + \frac{\E[R_m^4|\calE_m(\xi)]}{4}\left\|\E\left[\left(\mS\vvarphi_m\vvarphi_m^\T\mD + \mD\vvarphi_m\vvarphi_m^\T\mS\right)^2\right]\right\| \\
	&\leq \frac{\|\mD\|^2}{M^2} + \frac{K(K+1)}{M^2}\left\|\E\left[\left(\mS\vvarphi_m\vvarphi_m^\T\mD + \mD\vvarphi_m\vvarphi_m^\T\mS\right)^2\right]\right\|.
\end{align*}
Since $R_m$ and $\vvarphi_m$ are independent,
\[
	\E\left[\left(\mS\vvarphi_m\vvarphi_m^\T\mD + \mD\vvarphi_m\vvarphi_m^\T\mS\right)^2\right] =
	\frac{M^2}{4K(K+1)}
	\E\left[\left(\mS\vphi_m\vphi_m^\T\mD + \mD\vphi_m\vphi_m^\T\mS\right)^2\right].
\]
Expanding the square and making repeated use of Lemma~\ref{lm:EpptAppt} along with the facts,
\begin{align*}
	\|\mS\|& = \|\Ptha+\Pthb\| \leq 2, \\
	\trace(\mD\mS) &\leq \|\mD\|_F\|\mS\|_F \leq 
	\sqrt{2K}\|\mD\|\sqrt{2K}\|\mS\| \leq 4K\|\mD\|, \\
	\trace(\mD^2) &\leq 2K\|\mD\|^2, \\
	\trace(\mS^2) &\leq 8K,
\end{align*}
results in 
\[
	\left\|\E\left[\left(\mS\vphi_m\vphi_m^\T\mD + \mD\vphi_m\vphi_m^\T\mS\right)^2\right]\right\|
	~\leq~ \frac{32(K+1)}{M^2}\|\mD\|^2.
\]	
Thus
\begin{align*}
	\left\|\sum_{m=0}^M\E[\mX_m^2\mid\calE_m(\xi)]\right\| &\leq
	\frac{8K+9}{M}\|\mD\|^2.
\end{align*}
We now apply Proposition~\ref{prop:matrixbernstein} (matrix Bernstein) with the result above along with \eqref{eq:GincB} to get
\[
	\P{\|\mG_\tha-\mG_\thb\| > C\left(\sqrt{\frac{Kt}{M}}+\frac{\xi^2t}{M} \right)\|\mD\| ~\mid~ \calE(\xi)} ~\leq~
	4K\e^{-t}.
\]

Returning to \eqref{eq:condGinc}, we can combine the bound above along with the following direct consequence of Proposition~\ref{prop:gausslip}:
\begin{equation}
	\label{eq:GincmaxRm}
	\P{\sqrt{M}\max_{m=1,\ldots,M}R_m > \sqrt{2K} + \tau} ~\leq~ M\e^{-\tau^2/2}.
\end{equation}
For and $t,\tau\geq 0$ we have the unconditional bound,
\[
	\P{\|\mG_\tha-\mG_\thb\| > C\left(\sqrt{\frac{Kt}{M}}+\frac{(\sqrt{2K}+\tau)^2t}{M} \right)\|\mD\|}
	~\leq~
	4K\e^{-t} + M\e^{-\tau^2/2}.
\]
Taking $\tau = \sqrt{2t+2\log(M/K)}$ establishes the lemma.

\end{proof}

\subsection{Proof of Theorem~\ref{th:suphperp}}

We follow the same template as in Section~\ref{sec:th2proof}, but now the process $\Pth\mPhi^\T\mPhi\Pth^\perp\vh_0$ is vector-valued.  As before, we have
\[
	\sup_{\theta\in\Theta} \|\Pth\mPhi^\T\mPhi\Pth^\perp\vh_0\|_2 ~\leq~
	\|\mP_{\theta_0}\mPhi^\T\mPhi\mP_{\theta_0}^\perp\vh_0\|_2 ~+~
	\sup_{\theta\in\Theta} 
	\|\Pth\mPhi^\T\mPhi\Pth^\perp\vh_0- \mP_{\theta_0}\mPhi^\T\mPhi\mP_{\theta_0}^\perp\vh_0\|_2
\]
for any fixed $\theta_0$.  We can control the first term above using basic properties of Gaussian random vectors; a bound on the second term will follow from Proposition~\ref{prop:chaining} (with $c=0$) along with Lemma~\ref{lm:perphbound} below, which provides a tail bound on the increment of the process.

With $\theta_0$ fixed, the quantity $\mP_{\theta_0}\mPhi^\T\mPhi\mP_{\theta_0}^\perp\vh_0$ is a random vector formed by taking a random embedding of $\vh_0$ followed by an independent random projection.  Let $\mV$ be a $N\times K$ matrix with orthonormal columns such that $\mP_{\theta_0}=\mV\mV^\T$, and similarly let $\mU$ be a $N\times (N-K)$ orthonormal matrix with $\mP_{\theta_0}^\perp = \mU\mU^\T$.  Then
\begin{equation}
	\label{eq:Vtproj}
	\|\mP_{\theta_0}\mPhi^\T\mPhi\mP_{\theta_0}^\perp\vh_0\|_2 = \|\mV^\T\mPhi^\T\mPhi\mU\vg\|_2,
\end{equation}
where $\vg = \mU^\T\vh_0$; note that $\|\vg\|_2\leq\|\vh_0\|_2=1$.  Since $\mV^\T\mU = \mzero$, the $M\times K$ Gaussian random matrix $\mPhi\mV$ and the $M\times(N-K)$ Gaussian random matrix $\mPhi\mU$ are independent of one another.  Using a standard concentration bound (Proposition~\ref{prop:randomprojfixed} in Section~\ref{sec:probtools}), we have
\[
	\P{\|\mV^\T\mPhi^\T\mPhi\mU\vg\|^2_2 > \frac{K}{M}\left(1 + \max\left(\frac{7u}{K},\sqrt{\frac{7u}{K}}\right)\right)\|\mPhi\mU\vg\|_2} ~\leq ~
	\e^{-u},
\]
and
\[
	\P{\|\mPhi\mU\vg\|_2^2 > 1 + \max\left(\frac{7t}{M},\sqrt{\frac{7t}{M}}\right)} ~\leq~
	\e^{-t}.
\]
Taking $u=t$ in \eqref{eq:Vtproj} yields
\[
	\P{\|\mP_{\theta_0}\mPhi^\T\mPhi\mP_{\theta_0}^\perp\vh_0\|_2 > \gamma(t)} ~\leq~
	2\e^{-t},
	\quad
	\gamma(t) = \sqrt{\frac{K}{M}}\left(1 + \max\left(\frac{7t}{K},\sqrt{\frac{7t}{K}}\right)\right).
\]

For the increment bound, we use the following lemma.
\begin{lemma}
	\label{lm:perphbound}
	Let $\vh_0\in\R^N$ be a fixed vector with $\|\vh_0\|_2=1$.  There exists a constant $C$ such that for any fixed $\theta_1,\theta_2\in\Theta$,
	\begin{equation}
		\label{eq:perphincrement}
		\P{\|\Ptha\mPhi^\T\mPhi\Ptha^\perp\vh_0 - \Pthb\mPhi^\T\mPhi\Pthb^\perp\vh_0\|_2 > C\gamma(t)\|\Ptha-\Pthb\|}
		~\leq~
		4K\e^{-t},
	\end{equation}
	where
	\[
		\gamma(t) = \sqrt{\frac{tK}{M}} + \frac{t\sqrt{K}}{M}.
	\]
\end{lemma}

\begin{proof}

We can write $\Ptha\mPhi^\T\mPhi\Ptha^\perp\vh_0 - \Pthb\mPhi^\T\mPhi\Pthb^\perp\vh_0$ as the following sum of independent random vectors,
\[
	\Ptha\mPhi^\T\mPhi\Ptha^\perp\vh_0 - \Pthb\mPhi^\T\mPhi\Pthb^\perp\vh_0
	~=~
	\sum_{m=1}^M (\Ptha\vphi_m\vphi_m^\T\Ptha^\perp - \Pthb\vphi_m\vphi_m^\T\Pthb^\perp)\vh_0
	~=:~
	\sum_{m=1}^M \vx_m.
\]
The $\vx_m$ above are iid; for convenience, we define	
\begin{equation}
	\label{eq:xdef}
	\vx = \frac{1}{2}\left(\mD\vphi\vphi^\T\mT\vh_0 - \mS\vphi\vphi^\T\mD\vh_0\right),
\end{equation}
where $\vphi\sim\mathrm{Normal}(\vzero,\mId)$ and
\begin{align*}
	\mD = \Ptha - \Pthb = \Pthb^\perp-\Ptha^\perp, 
	\quad
	\mS = \Ptha + \Pthb,
	\quad
	\mT = \Ptha^\perp + \Pthb^\perp.
\end{align*}
Note that $M\E[\|\vx_m\|_2^2] = \E[\|\vx\|_2^2]$.  We have
\begin{align}
	\label{eq:xnorm}
	\|\vx\|_2^2 
	&= \frac{1}{4} \vh_0^\T\left(
	\mT\vphi\vphi^\T\mD^2\vphi\vphi^\T\mT - 
	2\mT\vphi\vphi^\T\mD\mS\vphi\vphi^\T\mD +
	\mD\vphi\vphi^\T\mS^2\vphi\vphi^\T\mD \right)\vh_0.
\end{align}
We will treat each of these three terms separately through applications of Lemma~\ref{lm:EpptAppt}.  For the first term, we apply
\begin{align*}
	\E[\vh_0^\T\mT\vphi\vphi^\T\mD^2\vphi\vphi^\T\mT\vh_0] &=
	\vh_0^\T\mT(2\mD^2 + \trace(\mD^2)\mId)\mT\vh_0 \\
	&\leq 2\|\mD\|^2\|\mT\vh_0\|_2^2 + \trace(\mD^2)\|\mT\vh_0\|_2^2 \\
	&\leq 8(K+1)\|\Ptha-\Pthb\|^2,
\end{align*}
since $\|\mT\vh_0\|_2^2\leq 4\|\vh\|_2^2 = 4$.  For the second term in \eqref{eq:xnorm},
\begin{align*}
	\left|\E[\vh_0^\T\mT\vphi\vphi^\T\mD\mS\vphi\vphi^\T\mD\vh_0]\right| &=
	\left|\vh_0^\T\mT(\mD\mS + \mS\mD + \trace(\mD\mS)\mId)\mD\vh_0\right| \\
	&\leq \left(2\|\mS\|\,\|\mD\| + \trace(\mD\mS)\right)\|\mT\vh_0\|_2\|\mD\vh_0\|_2 \\
	&\leq 8(K+1)\|\Ptha-\Pthb\|^2.
\end{align*}
Finally for the third term in \eqref{eq:xnorm},
\begin{align*}
	\E[\vh_0^\T\mD\vphi\vphi^\T\mS^2\vphi\vphi^\T\mD\vh_0] &=
	\vh_0^\T\mD(2\mS^2 + \trace(\mS^2)\mId)\mD\vh_0 \\
	&\leq (2\|\mS\|^2 + 8K)\|\mD\vh_0\|_2^2 \\
	&\leq 8(K+1)\|\Ptha-\Pthb\|^2.
\end{align*}
Collecting these results, we can take
\begin{equation}
	\label{eq:hperpsigma}
	\sigma^2 
	~=~ 
	\left\|\sum_{m=1}^M\E[\vx_m^\T\vx_m]\right\| 
	~=~
	M\E[\|\vx_m\|_2^2]
	~=~
	\frac{1}{M}\E[\|\vx\|_2^2]
	~\leq~
	\frac{8(K+1)}{M}\,\|\Ptha-\Pthb\|^2.
\end{equation}

Next we need to bound the $\psi_1$ norm of the $\vx_m$; we start with
\[
	\|\vx_m\|_{\psi_1} = \frac{\|\vx\|_{\psi_1}}{M}.
\]
We will bound each part of $\vx$ in \eqref{eq:xdef} separately.  For the first term, we note that $\vphi^\T\mT\vh_0$ is a Gaussian random scalar with variance $\|\mT\vh_0\|_2^2$, and so
\[
	\P{|\vphi^\T\mT\vh_0| > u} 
	~\leq~ 
	\exp\left(-\frac{u^2}{2\|\mT\vh_0\|_2^2}\right) 
	~\leq~
	\exp\left(-\frac{u^2}{8}\right).
\] 
Also, $\mD\vphi$ is a Gaussian random vector, and so we have the (standard) bound
\begin{align*}
	\P{\|\mD\vphi\|_2 > u} &\leq~ 2\exp\left(-\frac{u^2}{8\|\mD\|_F^2}\right) 
	~\leq~ 2\exp\left(-\frac{u^2}{16K\|\Ptha-\Pthb\|^2}\right).
\end{align*}
Thus for any $t,u\geq 0$,
\begin{align*}
	\P{\|(\vphi^\T\mT\vh_0)\mD\vphi\|_2 > t} 
	&\leq \exp\left(-\frac{u^2}{8}\right) +
	2\exp\left(-\frac{t^2}{u^2 16K\|\Ptha-\Pthb\|^2}\right).
\end{align*}
Taking $u=t^{1/2}(2K\|\Ptha-\Pthb\|^2)^{-1/4}$ above yields
\begin{align*}
	\P{\|(\vphi^\T\mT\vh_0)\mD\vphi\|_2 > t} 
	&\leq 3\exp\left(-\frac{t}{8\sqrt{2K}\|\Ptha-\Pthb\|}\right),
\end{align*}
and so
\begin{align}
	\label{eq:xpsi1a}
	\|(\vphi^\T\mT\vh_0)\mD\vphi\|_{\psi_1} &\leq 32\sqrt{2K}\|\Ptha-\Pthb\|.
\end{align}
We use similar computations for the second part of $\vx$ in \eqref{eq:xdef}:
\begin{align*}
	\P{|\vphi^\T\mD\vh_0| > u} &\leq 
	\exp\left(-\frac{u^2}{2\|\Ptha-\Pthb\|^2}\right),
	\quad
	\P{\|\mS\vphi\|_2 > u} ~\leq~
	2\exp\left(-\frac{u^2}{64K}\right),
\end{align*}
and so
\begin{align*}
	\P{\|(\vphi^\T\mD\vh_0)\mS\vphi\|_2 > t} &\leq 
	\P{|\vphi^\T\mD\vh_0| > \frac{\sqrt{\|\Ptha-\Pthb\| t}}{2(2K)^{1/4}}} +
	\P{\|\mS\vphi\|_2 > \frac{2(2K)^{1/4}\sqrt{t}}{\sqrt{\|\Ptha-\Pthb\|}}} \\
	&\leq 3\exp\left(-\frac{t}{8\sqrt{2K}\|\Ptha-\Pthb\|}\right),
\end{align*}
and so
\begin{align}
	\label{eq:xpsi1b}
	\|(\vphi^\T\mD\vh_0)\mS\vphi\|_{\psi_1} &\leq 32\sqrt{2K}\|\Ptha-\Pthb\|.
\end{align}

Combining \eqref{eq:xpsi1a} and \eqref{eq:xpsi1b} gives us the upper bound 
\begin{align*}
	B ~=~ 
	\max_{m}\|\vx_m\|_{\psi_1} 
	~=~
	\|\vx\|_{\psi_1}
	~\leq~
	\frac{64\sqrt{2K}}{M}\,\|\Ptha-\Pthb\|.
\end{align*}
The lemma follows by applying Proposition~\ref{prop:matrixbernstein} with this $B$ and $\sigma^2$ as in \eqref{eq:hperpsigma}.  Since the $\vx_m$ are all in the union of the column spaces of $\Ptha$ and $\Pthb$, we can take $K_1=2K$ and $K_2=1$ in Proposition~\ref{prop:matrixbernstein}.

\end{proof}

\section{Probability Tools}
\label{sec:probtools}


The following lemma was useful in establishing the increment bounds in Lemmas~\ref{lm:ProjectIncrement} and \ref{lm:perphbound}.  It can be established through a straightforward calculation.
\begin{lemma}
	\label{lm:EpptAppt}
	Let $\vphi\in\R^N$ be a Gaussian random vector, $\vphi\sim\mathrm{Normal}(0,\mId)$, and let $\mA$ be an arbitrary $N\times N$ matrix.  Then
	\[
		\E\left[\vphi\vphi^\T\mA\vphi\vphi^\T\right] = \mA + \mA^\T + \trace(\mA)\,\mId.
	\]
\end{lemma}


We also make use of the following standard concentration results for the norm of a Gaussian random vector.  Nice expositions of these facts can be found in (for example) \cite[Prop.\ 2.18]{ledoux2001concentration} or \cite[A.2.1]{vandervaart96we}.
\begin{proposition}
	\label{prop:gausslip}
	Let $\vx\in\R^K$ be distributed $\vx\sim\mathrm{Normal}(\mzero,\mId)$, and let $\mA$ be a matrix with $N$ columns.  Then
	\[
		\P{\|\mA\vx\|_2 \geq \|\mA\|_F + t} ~\leq~
		\exp\left(-\frac{t^2}{2\|\mA\|^2}\right).
	\]
	and 
	\[
		\P{\|\mA\vx\|_2 \geq \lambda} ~\leq~
		2\exp\left(-\frac{\lambda^2}{8\|\mA\|_F^2}\right).
	\]
\end{proposition}

Closely related is the following concentration result for the norm of a random matrix applied to a fixed vector.
\begin{proposition}
	\label{prop:randomprojfixed}
	Let $\vx\in\R^N$ be a fixed vector with $\|\vx\|_2=1$, and let $\mPhi$ be an $M\times N$ matrix whose entries are standard independent Gaussian random variables.  Then
	\[
		\P{\|\mPhi\vx\|_2^2 > M + \max\left(C t,\sqrt{CM t}\right)} ~\leq~\e^{-t},
	\]
	for a constant $C \leq 2/\log(\e/2)\approx 6.52$.
\end{proposition}
\begin{proof}
	The quantity $\|\mPhi\vx\|_2^2$ is a sum of $M$ independent chi-square random variables.  A standard Chernoff bound (see, for example, \cite[Chapter 1]{vempala04ra}) gives us
	\begin{align*}
		\P{\|\mPhi\vx\|_2^2 > M(1 + \lambda)} &\leq ((1+\lambda)\e^{-\lambda})^{M/2}.
	\end{align*}
	The proposition then follows by applying the inequality
	\[
		(1+\lambda)\e^{-\lambda}~\leq~ \e^{-\min(\lambda,\lambda^2)\log(\e/2)},
		\quad \lambda\geq 0,
	\]
	and taking $t=\min(\lambda,\lambda^2)M/C ~\Rightarrow~ M\lambda = \max(Ct,\sqrt{CMt})$ with $C=2/\log(\e/2)$.
\end{proof}

The Matrix Bernstein inequality, stated below, gives us a tight way to estimate the size of a sum of independent random matrices.  An early result for this type of bound appeared in \cite{ahlswede02st} and was greatly refined in \cite{tropp12us}.
\begin{proposition}[Matrix Bernstein]
	\label{prop:matrixbernstein}
	Let $\mX_1,\ldots,\mX_M$ be independent self-adjoint random $K\times K$ matrices with $\E[\mX_m]=\mzero$ and
	\[
		\|\mX_m\|\leq B.
	\]
	Then there exists a universal constant $C_B\leq 4$ such that for all $t\geq 0$,
	\begin{equation}
		\label{eq:bernbound}
		\P{\left\|\sum_{m=1}^M\mX_m\right\|\geq C_B\gamma(t)}
		~\leq~ 2K\e^{-t},
	\end{equation}
	where
	\begin{equation}
		\gamma(t) = \max\left\{\sigma\sqrt{t},~2Bt\right\},
		\qquad
		\sigma^2 = \left\|\sum_{m=1}^M\E[\mX_m^2]\right\|.
	\end{equation}
	For non-symmetric $K_1\times K_2$ matrices $\mX_m$, we have the similar bound
	\[
		\P{\left\|\sum_{m=1}^M\mX_m\right\|\geq C_B\gamma(t)}
		~\leq~ 2(K_1+K_2)\e^{-t},
	\] 
	for the same $\gamma(t)$ with
	\[
		\sigma^2 =
		\max\left\{\left\|\sum_{m=1}^M\E[\mX_m^\T\mX_m]\right\|,
		\left\|\sum_{m=1}^M\E[\mX_m\mX_m^\T]\right\|\right\}.
	\]
\end{proposition}

In our proof of Theorem~\ref{th:suphperp}, we find it slightly more convenient to use a version of the matrix Bernstein inequality for matrices that have norms which are subexponential (rather than bounded with probability $1$).  A convenient way to characterize this property is using the Orlicz norm.
\begin{definition}
	The Orlicz-$\psi_1$ norm of a positive random variable $X$ is defined as
\[
	\|X\|_{\psi_1} = \inf_{u\geq 0}\left\{ u ~:~ \E\left[\e^{X/u}\right]\leq 2\right\}.
\]
For a random vector $\vx$, we write $\|\vx\|_{\psi_1}$ for the Orlicz-$\psi_1$ norm of $\|\vx\|_2$; for a random matrix $\mX$, we write $\|\mX\|_{\psi_1}$ for the Orlicz-$\psi_1$ norm of its operator norm $\|\mX\|$.
\end{definition}
If $X$ is subexponential,
\begin{equation}
	\label{eq:subexp}
	\P{X > u} ~\leq~ \alpha\,\e^{-\beta u},
	\quad\text{for some } \alpha,\beta>0,
\end{equation}
then it has finite $\psi_1$ norm.  In fact, a straightforward calculation shows that for $X$ obeying \eqref{eq:subexp},
\[
	\E\left[\e^{X/u}\right] 
	~\leq~
	\left(1+\frac{1}{\beta u - 1}\right)\alpha^{1/\beta u},
\]
and so
\begin{equation}
	\label{eq:psi1bounds}
	\|X\|_{\psi_1} ~\leq~
	\begin{cases}
		\frac{\alpha+1}{\beta} & \text{for all}~\alpha\geq 1, \\
		\frac{2\log\alpha + 1}{\beta} &\text{for all}~\alpha\geq 3.6,\\
		\frac{2\log\alpha}{\beta} &\text{for all}~ \alpha \geq 18.
	\end{cases}
\end{equation}

In Proposition \ref{prop:matrixbernstein}, we can replace the condition that the $\|\mX\|\leq B$ with the condition that the $\psi_1$-norms can be uniformly bounded,
\[
	\|\mX_m\|_{\psi_1}\leq B, \quad m = 1,\ldots,M.
\]
Then \eqref{eq:bernbound} holds with (see \cite[Sec.\ 2.2]{koltchinskii12vo})
\begin{equation}
	\label{eq:psi1berngamma}
	\gamma(t) = \max\left\{\sigma\cdot\sqrt{t},~2B\log(2B\sqrt{M}/\sigma)\cdot t\right\},
\end{equation}



\newpage

\appendix


\section{Proofs of Lemmas~\ref{lm:sobolev_shift} and \ref{lm:tv_shift}}

Assume without loss of generality that $\theta_2\geq\theta_1$, set $\theta_d = \theta_2-\theta_1$, and notice that
\[
	\|v(t-\theta_1) - v(t-\theta_2) \|_2^2 =  \int_{-\infty}^{\infty} |v(t) - v(t-\theta_d)|^2 dt.
\]
Lemma~\ref{lm:sobolev_shift} follows directly from the fact that the Fourier transform preserves energy:
\begin{align*}
	\int_{-\infty}^{\infty} |v(t) - v(t-\theta_d)|^2 dt		
	&= \frac{1}{2\pi}\int_{-\infty}^\infty|\hat{v}(\xi)|^2\,
	|1-e^{-j\xi\theta_d}|^2~d\xi \\
	&\leq \frac{\theta_d^2}{2\pi}\int_{-\infty}^\infty|\hat{v}(\xi)|^2\,
	|\xi|^2~d\xi.
\end{align*}
For Lemma~\ref{lm:tv_shift}, we separate the integral into intervals of length $\theta_d$:
\begin{align*}
	\int_{-\infty}^{\infty} |v(t) - v(t-\theta_d)|^2 dt
	&= \int_0^{\theta_d} \sum_{k=-\infty}^{\infty} |v(t+k\theta_d) - v(t+(k-1)\theta_d)|^2 \\ 
	&\leq \int_0^{\theta_d} \left(\sum_{k=-\infty}^{\infty}|v(t+k\theta_d) - v(t+(k-1)\theta_d)| \right)^2 \\
	& \leq \theta_d\,\|v\|_{TV}^2 .
\end{align*}

\section{Proof of Lemma~\ref{lm:gsum}}

We have
\begin{align*}
	\sum_{j\geq 0} \frac{1}{2}2^{-j}g(u_j) &=
	\sum_{j\geq 0} \frac{1}{2}2^{-j}\left(a\sqrt{u_j} + bu_j + cu_j^2\right).
\end{align*}
Since $a\sqrt{u_j} + bu_j$ is a concave function of $u_j$ and since $\sum_{j\geq 0}2^{-j-1}=1$ and also $\sum_{j\geq 0}j2^{-j-1}=1$, by the Jensen inequality
\begin{align}
	\nonumber
	\sum_{j\geq 0} \frac{1}{2}2^{-j}\left(a\sqrt{u_j} + bu_j\right) 
	&\leq
	a\sqrt{\sum_{j\geq 0}\frac{1}{2}2^{-j}u_j } + b\sum_{j\geq 0}\frac{1}{2}2^{-j}u_j \\
	\label{eq:lu1}
	&= a\sqrt{u+p+q} + b(u+p+q).
\end{align}
Now using the fact that $\sum_{j\geq 0}j^2 2^{-j-1} = 3$,
\begin{align}
	\nonumber
	\sum_{j\geq 0} \frac{1}{2}2^{-j} cu_j^2 &=
	c\sum_{j\geq 0} \frac{1}{2}2^{-j}(u^2+p^2j^2+q^2 + 2upj + 2uq + 2pqj) \\
	\nonumber
	&= c(u^2 + 3p^2 + q^2 + 2up + 2uq + 2pq) \\
	\label{eq:lu2}
	&= c(u + p + q)^2 + 2cp^2.
\end{align}
The lemma is proved by combining \eqref{eq:lu1} and \eqref{eq:lu2}.


{\small
\bibliographystyle{plain}
\bibliography{csm-refs}
}

\end{document}